\documentclass[paper]{article}

\usepackage{amsmath, amsfonts, amssymb, amsthm, latexsym}
\usepackage{authblk}
\usepackage{todonotes}
\usepackage{color}

\usepackage{bm}
\usepackage{graphicx}
\usepackage{tikz}
\usepackage{subfigure}
\usepackage{comment}

\let\Re\undefined
\let\Im\undefined
\DeclareMathOperator{\Re}{Re}
\DeclareMathOperator{\Im}{Im}

\DeclareMathOperator{\supp}{supp}

\newcommand{\ds}{\displaystyle}

\newtheorem{theorem}{Theorem}[section]
\newtheorem{lemma}[theorem]{Lemma}
\newtheorem{proposition}[theorem]{Proposition}
\newtheorem{corollary}[theorem]{Corollary}
\theoremstyle{definition}
\newtheorem{definition}[theorem]{Definition}
\newtheorem{example}{Example}
\newtheorem{remark}[theorem]{Remark}
\newtheorem{assumption}[theorem]{Assumption}

\numberwithin{equation}{section}

\definecolor{shadecolor}{rgb}{0.95, 0.95, 0.86}
\usetikzlibrary{decorations.markings}

\def\be{\begin{equation}}
\def\ee{\end{equation}}
\def\bi{\begin{itemize}}
\def\ei{\end{itemize}}

\def\bea{\begin{eqnarray}}
\def\eea{\end{eqnarray}}
\def\bl{\begin{lemma}}
\def\el{\end{lemma}}
\def\bd{\begin{definition}}
\def\ed{\end{definition}}
\def\bp{\begin{proposition}}
\def\ep{\end{proposition}}

\def\br{\begin{remark}}
\def\er{\end{remark}}

\def\bt{\begin{theorem}}
\def\et{\end{theorem}}
\def\bc{\begin{corollary}}
\def\ec{\end{corollary}}

\def\ra{\rightarrow}

\newcommand{\G}{\Gamma}
\renewcommand{\O}{\Omega}
\renewcommand{\k}{\varkappa}
\renewcommand{\d}{\delta}

\newcommand{\e}{\epsilon}
\renewcommand{\o}{\omega}
\newcommand{\g}{\gamma}

\renewcommand{\part}{\partial}

\newcommand{\Res}{\text{Res} \,}

\def\Rscr{\mathcal R}

\newcommand{\sech}{\hbox{sech}}

\def\le{\left}
\def\ri{\right}

\def\C{{\mathbb C}}
\def\R{{\mathbb R}}
\def\N{{\mathbb N}}

\def\a{\alpha}

\def\g{\gamma}

\def\m{\mu}

\def\n{ {\nu}}

\def\1{{\bf 1}}
\def\r{\rho}
\def\s{ {\sigma}}
\def\t{ {\tau}}

\def\z{\zeta}

\def\hf{\frac{1}{2}}

\begin{document}

\title{On minimal energy solutions to certain classes of integral equations related to soliton gases for integrable 
systems}
\author[1]{Arno Kuijlaars}
\author[2]{Alexander Tovbis}

\affil[1]{Katholieke Universiteit Leuven, Belgium, arno.kuijlaars@kuleuven.be}
\affil[2]{University of Central Florida, Orlando FL, U.S.A.,
Alexander.Tovbis@ucf.edu}

\footnotetext{The work of the first author is supported by long term structural funding-Methusalem grant
of the Flemish Government, by
the Fonds Wetenschappelijk Onderzoek – Vlaanderen (FWO) and the Fonds de la Recherche Scientifique – FNRS under EOS Project No. 30889451,
and by FWO Flanders projects G.0864.16 and G.0910.20.
 The work of the second author is supported  by NSF grant DMS-2009647.}

\maketitle
\tableofcontents

\begin{abstract} 
We prove existence, uniqueness and non-negativity of solutions of certain integral equations 
describing the density of states $u(z)$ 
in the spectral theory of soliton  gases 
for the one dimensional integrable focusing Nonlinear  Schr\"{o}dinger Equation (fNLS) and for 
the Korteweg de Vries (KdV) equation. Our proofs are based on
 ideas and methods of potential theory. In particular, we show that the minimizing (positive) measure for 
 certain energy functional is absolutely continuous and its density $u(z)\geq 0$ solves the required integral equation. In a similar fashion we show that $v(z)$, the temporal analog of $u(z)$, is the difference of densities of two absolutely continuous measures. Together, integral equations for $u,v$ represent nonlinear dispersion relation
for the fNLS soliton gas.
We also discuss smoothness and other properties of the obtained solutions. 
Finally, we obtain exact solutions of the above integral equations in the case of a KdV condensate and a bound state 
fNLS condensate.
Our results is a first step 
towards a mathematical foundation for the spectral  theory of soliton and breather gases, which 
appeared in work of El and Tovbis, Phys. Rev. E, 2020.
It is expected that the presented ideas and methods will 
be useful for studying similar classes of integral equation describing, for example, breather gases for the fNLS, 
as well as  soliton gases of various
integrable systems.
\end{abstract}

\section{Introduction and statement of results}\label{intro}
\subsection{Introduction}
Let $\mathbb C^+$ denote the upper half-plane and $\Gamma^+ \subset \mathbb C^+ \cup \mathbb R$
be a compact and let  $\sigma : \Gamma^+ \to [0,\infty)$ 
be a continuous, non-negative function on $\Gamma^+$.
The motivation of this paper are two independent integral equations
\begin{align} \label{dr_soliton_gas1}
 \frac 1{\pi} \int _{\Gamma^+}\log \left|\frac{w-\bar z}{w-z}\right|
u(w)d\lambda(w)+\sigma(z)u(z)& = \Im z,\\
\label{dr_soliton_gas2}
\frac 1{\pi} \int _{\Gamma^+}\log \left|\frac{w-\bar z}{w-z}\right| v(w) d\lambda(w)+  \sigma(z) v(z)& = -4 \Im z\Re z,
\end{align}
for unknown functions $u$ and $v$ respectively, where $z\in\Gamma^+$ 
and $\lambda$ is some reference measure on $\Gamma^+$. 
For example, $\lambda$ could be the area measure
in a $2D$ context, or the arclength measure in the case of a contour $\Gamma^+$.
The exact meaning of $\lambda$ will be discussed 
in Assumption \ref{assumptiondm} below.

Our goal is to prove the existence and uniqueness
of such solutions and, what is especially important, the fact that the solution $u$ of \eqref{dr_soliton_gas1}
satisfies $u(z)\geq 0$ everywhere on $\Gamma^+$. This property of $u$ is natural from the interpretation 
of $u$ as ``density of states'' in the soliton gas theory, 
that is, the average number of ``nonlinear waves''
with given spectral characteristics
per unit of length and per unit of ``measure'' 
on $\Gamma^+$.
Thus, the present paper is a significant step towards the mathematical foundation of the spectral theory of soliton gases
for the focusing  Nonlinear Schr\"{o}dinger equation (fNLS) that was recently presented in \cite{ElTovbis},
as well as for the Korteweg de Vries equation (KdV) that was first presented in \cite{El2003}.

A brief description of how equations \eqref{dr_soliton_gas1}-\eqref{dr_soliton_gas2},
which we will call nonlinear dispersion relation (NDR),
appear in the spectral theory for the fNLS soliton gas  will be given in Section \ref{sect-backg}.
We will also consider the case of a more general right hand side in \eqref{dr_soliton_gas1} that we will 
denote  by $\varphi(z)$. 
Finally, we are also interested in the support of $u,v$ 
and the smoothness of 
$u,v$ under various assumptions on the smoothness of $\sigma$ and the geometry of $\Gamma^+$. 

If $\sigma>0$ on $\Gamma^+$ then  equations \eqref{dr_soliton_gas1}-\eqref{dr_soliton_gas2} are  Fredholm integral equations 
of the second kind.  In the case $\sigma \equiv 0$ on $\Gamma^+$  equations \eqref{dr_soliton_gas1}-\eqref{dr_soliton_gas2} are  
Fredholm integral equations 
of the first kind and the general case  $\sigma \geq 0$ on $\Gamma^+$ is sometimes called Fredholm integral equations 
of the third kind \cite{Prossdorf}. Whereas  there exists  well known theory for second kind Fredholm  equations
that we can use to prove the existence  and uniqueness
of  $u(z)$ when $\sigma>0$ on $\Gamma^+$, the difficulty still lies in proving that the obtained $u(z)\geq 0$ on $\Gamma^+$.
However,  when it comes to the general case  $\sigma \geq 0$, much less is known even about the existence of $u(z)$.

We study the NDR equations with potential theory for
the upper half-plane, as the function 
\begin{equation} 
	\frac{1}{\pi} \int _{\Gamma^+}\log \left|\frac{w-\bar z}{w-z}\right| u(w)d \lambda(w),
\end{equation}
defines the Green potential for the upper half-plane $\mathbb C^+$ of the measure 
\begin{equation} \label{measure}
d\mu =u(z)d \lambda(z).
\end{equation}
Then both equations  \eqref{dr_soliton_gas1}-\eqref{dr_soliton_gas2}   can be written as 
\begin{equation} \label{Gueq} 
G \mu + \sigma u = \varphi \qquad \text{ on } \Gamma^+, 
\end{equation}
where $\varphi(z)$ coincides with either $\Im z$ or with $-4 \Im z \Re z = -2\Im (z^2) $ respectively, and
we also write 
\begin{equation}  \label{Greenpot}
G\mu (z) = \frac{1}{\pi} \int \log \left| \frac{z-\overline{w}}{z-w} \right|
d\mu(w) 
\end{equation}
for the Green potential of a positive Borel measure  $\mu$ in $\mathbb C^+$. 
Note that $G\mu$ is superharmonic
on $\mathbb C^+$, harmonic on $\mathbb C^+ \setminus \supp(\mu)$
and $G\mu = 0$ on the real line and at infinity (provided that
$\mu$ has compact support in $\mathbb C^+$).

Consider first the case of  $\sigma \equiv 0$ on $\Gamma^+$, which corresponds to the soliton condensate, 
\cite{ElTovbis}. Then \eqref{Gueq} becomes 
\begin{equation} \label{cond}
G\mu=\varphi.
\end{equation}
Our first observation is that \eqref{cond} is the Euler-Lagrange equation for the Green energy functional  
\begin{equation} \label{J0}
J_0(\mu) = \int G(\mu) d\mu-2\int\varphi d\mu, 
\end{equation}
which we want to minimize among all the Borel measures $\mu$ with $\supp\mu \subset \Gamma^+$. 
It is well known \cite{Helms} \cite{SaffTotik},  that (if $\Gamma^+$ is a compact 
subset of $\mathbb C^+$ of positive capacity, 
and $\varphi$ is continuous) the 
minimizing measure $\mu^*$ exists and is unique. Moreover,   $\mu^*$ satisfies equation \eqref{cond} quasi everywhere (q.e.),
{i.e., up to a possible set of zero capacity,}
on  $\supp\mu^*$ and the inequality
\begin{equation} \label{ineq}
G\mu^* \geq \varphi
\end{equation}
holds q.e.\ on $\Gamma^+\setminus\supp\mu^*$. Thus, our goals are: 
\begin{itemize}
\item  to modify the energy functional $J_0$ so that the corresponding  Euler-Lagrange equation for the minimizer $\mu^*$
will be \eqref{Gueq} instead of \eqref{cond};
\item to  prove  that under suitable conditions on $\varphi$ the Euler-Lagrange equation \eqref{Gueq} for $\mu^*$ holds not only on $\supp \mu^*$ but on the full $\Gamma^+$;
\item 
to clarify the meaning of the density $u^*$ of $\mu^*$ 
with respect to the reference measure $\lambda$.
\end{itemize}

\subsection{Minimization of modified energy functional}
We are able to satisfactorily answer these questions
in case $\Gamma^+ \subset \mathbb C^+$. Additional
complications arise in case $\Gamma^+ \cap 
\mathbb R \neq \emptyset$ (which is a very relevant
situation)
that we cannot overcome at this moment. 
Thus we restrict to $\Gamma^+$ being a compact
subset of $\mathbb C^+$.
The measure $\lambda$ is assumed to satisfy the
following mild condition.
\begin{assumption} \label{assumptiondm}
	We assume that $\supp(\lambda) = \Gamma^+$ with 
	$0 < \int_{\Gamma^+} d\lambda < +\infty$, and
	its Green potential $G(\lambda)$ 
	is bounded
	and continuous on $\mathbb C^+$.
\end{assumption}
It follows from Assumption \ref{assumptiondm} that
$\Gamma^+$ has positive logarithmic 
capacity \cite{SaffTotik}. 
As typical examples we may think of $\Gamma^+$ 
as a finite union of {piece-wise} smooth contours and
closed 2D regions (closure of connected open set), 
where $\lambda$ is arclength measure
on a smooth contour and $\lambda$
is the Lebesgue area measure on a 2D domain.

Now we introduce the energy functional that is
a modification of \eqref{J0}. 
\begin{definition}  \label{def:Jsigma} 
	For a continuous $\varphi : \Gamma^+ \to \mathbb R$ we
	define
	\begin{equation} \label{Jmu}	
	J_{\sigma}(\mu) 
	:= \begin{cases} \ds J_0(\mu)  + \int \sigma u^2 d\lambda,
	\quad & \begin{array}{ll} \text{if $\sigma \mu = \sigma u \lambda$ is absolutely} \\
	\text{continuous
		with respect to $\lambda$},
	\end{array} \\
	+\infty, & \text{ otherwise}.
	\end{cases}
	\end{equation}
\end{definition}

The first main result of the paper is the following 
Theorem \ref{solitonmain}, which  is proven 
in Section \ref{proofmain}.

\begin{theorem} \label{solitonmain}
Let $\Gamma^+ \subset \mathbb C^+$ be a  
compact with a measure $\lambda$ that satisfies 
Assumption \ref{assumptiondm}.
Suppose the functions  $\varphi : \Gamma^+ \to \mathbb R$ and  
$\sigma : \Gamma^+ \to [0,\infty)$
are continuous. Then the following hold.
\begin{enumerate}
	\item[\rm (a)]
	There is unique minimizing measure $\mu^*$ on
	$\Gamma^+$ 
	for the energy 	functional $J_{\sigma}$ that is defined in Definition \ref{def:Jsigma}.
	The measure	$\sigma \mu^*$ is absolutely continuous
	with respect to $\lambda$, that is, $\sigma u^* \lambda = \sigma \mu^*$ for some density $u^*\in L^1(\sigma
	\lambda)$.
	\item[\rm (b)] 
	If $u^*$ is such
	that $\sigma u^* \lambda = \sigma \mu^*$, then we have
 	\begin{equation} \label{maineq-phi}
 G \mu^* + \sigma u^* = \varphi \qquad
 	\mu^*\text{-a.e. }  \text{ on } \Gamma^+.
 	\end{equation}

	\item[\rm (c)] If $\varphi$
	is defined everywhere on $\mathbb C^+$ and is positive,
	continuous, and superharmonic there, then also 
	\begin{equation}  \label{maineq-phi-outside}
	G\mu^* = \varphi  \quad \text{ on } 
		\Gamma^+ \setminus \supp(\mu^*), \end{equation}
	while $G \mu^* \leq \varphi$ on $\mathbb C^+$.
\end{enumerate}
\end{theorem}

The equation \eqref{maineq-phi} is the Euler-Lagrange
variational equality for the minimization problem.
Since $u^*$ is only defined $\mu^*$-a.e., it is in a way  natural
to have  \eqref{maineq-phi} only $\mu^*$-a.e.
The  equation \eqref{maineq-phi} is accompanied by  a 
variational inequality outside of the
support of $\mu^*$ (where $u^* = 0$) which says that
$G \mu^*  \geq \varphi$ on $\Gamma^+ \setminus \supp(\mu^*)$
up to a possible set of zero capacity.
The conditions on $\varphi$ in item (c) imply 
equality outside of the support,
and this includes the case $\varphi(z) =  \Im z$.
Under these conditions, the minimizer $\mu^*$ will be called  {\it weak solution} of the equation \eqref{Gueq}.

The arguments in the proof of Theorem \ref{solitonmain} lead to important characterization of $\supp\mu^*$.
 Let $\Omega$ denote the unbounded component of $\mathbb C^+\setminus \Gamma^+$. Then conditions (c), Theorem \ref{solitonmain} 
 together with some mild requirements on $\varphi(z)$ near infinity imply that
$\part \Omega \cap \Gamma^+ \subset \supp \mu^*$, see Proposition \ref{propGvareq2} for exact formulation. In particular, if $\Gamma^+$
is a collection of open  arcs (each arc has endpoints) then $\supp\mu^*=\Gamma^+$.

\begin{remark} \label{prop-cond}
In the assumptions of Theorem \ref{solitonmain} let us fix some $\G^+$ and $\varphi$
but allow $\s$ to vary. Then it is clear that
\be
J_0(\m_0^*)\leq J_0(\m_\s^*) \leq J_\s(\m_\s^*),
\ee
where $\m_\s^*$ denotes the minimizer of $J_\s$ for a given $\s$.
Thus, {\it the condensate $\s\equiv 0$ corresponds to the minimal energy} for given $\G^+$ and $\varphi$.
The converse statement, in general, is not true, as it follows from  
Example \ref{ex-circle} below.
We also observe that the condensate maximizes the value of  $\int_{\G^+}\varphi d\m^*_\s$ 
with  given $\G^+$ and $\varphi$ since, 
by the definition  \eqref{J0}, \eqref{Jmu}
of $J_{\sigma}$ 
\begin{align*}
	J_\s(\m_\s^*) & = \int_{\Gamma^+} G \mu_{\sigma}^*
		d\mu_{\sigma}^* - 2 \int_{\Gamma^+} \varphi d \mu_{\sigma}^*
		+ \int_{\Gamma^+} \sigma (u_{\sigma}^*)^2 d \lambda \\
		& = \int_{\Gamma^+} 
		\left( G \mu_{\sigma}^* + \sigma u_{\sigma}^*
		- 2 \varphi \right) d\mu_{\sigma}^* \\
	& = - \int_{\Gamma^+} \varphi d\mu^*_\sigma,
\end{align*}
where for the last line we used the identity 
\eqref{maineq-phi} that is valid a.e.\ on the support
of $\mu_{\sigma}^*$.  In particular, $\int_{\G^+}\varphi d\m^*_\s$ is related with the average
intensity of the fNLS soliton gas when $\varphi(z)=\Im z$, which is maximized in the case of the condensate.
\end{remark}

\subsection{Equality in variational condition}

In our second main result we give conditions that guarantee that the equation \eqref{maineq-phi}
is valid everywhere on $\supp(\mu^*)$ instead of being valid just $\mu^*$-a.e. 
We have two such conditions. The first condition is that $\sigma > 0$. Then it turns out
that $\mu^*$ has a continuous density as we show in part (a) of Theorem \ref{secondthm} and
\eqref{maineq-phi} is satisfied everywhere  on $\Gamma^+$
where $\sigma > 0$.

The second condition deals with the case when $\sigma = 0$
on $\Gamma^+$ or on part of $\Gamma^+$.
When $\sigma \equiv 0$ on $\Gamma^+$ then 
it is known from potential theory \cite{Ransford, SaffTotik} that 
the identity \eqref{maineq-phi} may fail on a subset $E$
of the support of $\mu^*$ of capacity zero. The
set $\Gamma^+$ is thin at the points in $E$ 
in the following sense, see \cite[Definition 3.8.1]{Ransford}.

\begin{definition} \label{thick}
	Let $S$ be a subset of $\mathbb C$ and let $z_0 \in 
	\mathbb C$. Then $S$ is thick (or non-thin) at
	$z_0$ if $z_0 \in \overline{S \setminus \{z_0\}}$ 
	and if, for every superharmonic function
	$u$ defined on a neighborhood of $z_0$,
	\[ \liminf_{z \to z_0 \atop z \in S \setminus \{ z_0\}} u(z) = u(z_0), \]
	Otherwise, $S$ is thin at $z_0$.
\end{definition}
Thus, in case $\sigma \equiv 0$ on $\Gamma^+$,
and $\Gamma^+$ is thick at all of its points, then
(1.10) holds on the full support of $\mu^*$,
and under the conditions of Theorem 1.2 (c),
the identity $G \mu^* = \varphi$ holds on the
full set $\Gamma^+$. 

A connected set with more than one point (for
example a contour) is thick at
all of its points. On the other hand, a countable set
is thin at every point. 

The notion of thickness is related 
to the solvability of the Dirichlet problem 
for harmonic functions.  
If $\Omega$ is a bounded open set, and $f$
is a continuous function on $ \partial \Omega$,
then the Dirichlet problem asks for a
continuous function $u$ on $\overline{\Omega}$ that
is harmonic in $\Omega$ and agrees with $f$ on
the boundary.  
The Dirichlet problem is solvable for
every continuous function on $\partial \Omega$ 
if and only if $\partial \Omega$ is thick at all
of its points, see e.g.\ \cite[Theorem 7.5.1]{ArmitageGardiner} or 
\cite[Appendix A.2, Theorem 2.1]{SaffTotik}. 

\begin{theorem} \label{secondthm}
	Under the general assumptions of Theorem \ref{solitonmain},
	let $\mu^*$ be the minimizer of $J_\sigma$ with density $u^*$.
\begin{enumerate}
	\item[\rm (a)] Let $S = \{ z \in \Gamma^+ \mid
	\sigma(z) > 0\}$. Then $G \mu^*$ is continuous
	on $S$, the density $u^*$ of $\mu^*$ is
	continuous on $S$ (after modifying it on
	a set of $\mu^*$-measure $0$, if necessary),
	and 
	\[ G \mu^* + \sigma u^* = \varphi \quad \text{ on } S.\]
	\item[\rm (b)] 
	Let $\varphi$ be positive, continuous, and
	superharmonic on $\mathbb C^+$.
	Let $\sigma$ be continuous on $\Gamma^+$, and
	$S_0 = \{ z \in \Gamma^+ \mid \sigma(z) = 0 \}$.
	Suppose $S_0$ is thick at $z_0 \in S_0$
	(see Definition \ref{thick}). 	Then
	\[ G\mu^*(z_0) = \varphi(z_0). \]
\end{enumerate}
\end{theorem}

Combining parts (a) and (b) of Theorem \ref{secondthm} we get the following.

\begin{corollary} \label{cor-eq-u}
	If the zero set $S_0$ of  $\sigma$ is thick at each
	of its points (and $\varphi$ is positive, continuous,
	and superharmonic) then
	\[ G\mu^* + \sigma u^* = \varphi \] 
	holds  on all of $\Gamma^+$. 
	This holds in particular if the zero set is empty,
	or if it is a connected set with more than one 
	point, or a union of such sets. 
\end{corollary}

In the case where the zero set of $\sigma$
has an isolated point, one 
may encounter the situation where $\sigma(a) = 0$
and $u^*(a) = \infty$ 
and then $\sigma u^*$ is not well defined at this
point. This happens in the following example.

\begin{example}\label{ex-Sigma-nontiv}
	Suppose $\Gamma^+$ is a  bounded smooth arc
	in $\mathbb C^+$ with arclength measure $\lambda$.
	Take a point $a \in \Gamma^+$ and consider
	\[ d \mu^*(z) = c |z-a|^{-1/2}  d \lambda(z),
		\qquad z \in \Gamma^+ \]
	with $c > 0$.
	This is a finite measure with a Green
	potential that is bounded and continuous.
	For small  enough $c >0$ we have $G \mu^* < \varphi$
	on $\Gamma^+$.  Then
	\[ \sigma(z) := c^{-1} |z-a|^{1/2} 
		\left(\varphi(z) - G \mu^*(z) \right),
	\qquad z \in \Gamma^+, \]
	is non-negative and continuous on $\Gamma^+$.
	
	The measure 	$\mu^*$ is the minimizer of $J_{\sigma}$ 
	for this $\sigma$, with the density
	\[ u^*(z) = c |z-a|^{-1/2}. \]
	By construction, the equality 
	$G \mu^* + \sigma u^* = \varphi$ holds on $\Gamma^* \setminus \{a\}$. At $z=a$ we have $\sigma(a) = 0$ and $u^*(a) = +\infty$
	and the product $\sigma u^*$ is not well-defined at $a$.
	
	To have the equality at $z=a$ as well, we
	need to interpret the product $\sigma(a) u^*(a)$ in
	this situation as $\varphi(a) - G\mu^*(a) > 0$.
\end{example}

Neither of Theorem \ref{solitonmain}, part (c), Theorem \ref{secondthm}, part (b) or  Corollary \ref{cor-eq-u} covers
the corresponding to \eqref{dr_soliton_gas2} case $\varphi(z) = -4 \Im z \Re z = -2 \Im (z^2)$,
 since this $\varphi(z)$, although  harmonic,  takes 
both positive and negative
values in $\mathbb C^+$. 
Nevertheless, in Theorem \ref{theo-eq2}, Section \ref{sect-sol-dr2}, we  construct the
 solution to \eqref{dr_soliton_gas2}  by representing the right hand side of   \eqref{dr_soliton_gas2} as a 
 difference of two continuous, positive and superharmonic in $\mathbb C^+$ functions, 
 to which we can apply the statements mentioned at the beginning of this paragraph.
 
 \subsection{Outline}

Here is a brief description of the rest of the paper. 
In  Section \ref{sect-backg} we give a concise  presentation of  the ideas leading to NDR  
\eqref{dr_soliton_gas1}-\eqref{dr_soliton_gas2}
for the fNLS soliton gas. In fact, we will obtain there the more general NDR \eqref{dr_breather_gas1}-\eqref{dr_breather_gas2}
for the fNLS breather gas, for which soliton gas is a particular case. We also  describe special cases of soliton gases,
such as fNLS bound state soliton gas and   fNLS soliton condensate. 
It is worth mentioning here that the methods of potential theory, used in this paper, can be applied to the breather
gas as well. The authors have obtained partial results in this direction that they hope to complete at a later time.
We also observe there that fNLS bound state soliton gas
can be seen as essentially equivalent to the KdV soliton gas, which is the first example of a soliton gas that was obtained
in \cite{El2003}. Thus, correspondingly modified Theorems~ \ref{solitonmain}, \ref{secondthm}, \ref{theo-eq2}, Corollary 
\ref{cor-eq-u} and  Propositions \ref{propGvareq2}, \ref{prop-cond}, as well as some results of Section \ref{sect-prop-mu},  
are applicable to the KdV soliton gas.

Sections \ref{proofmain}-\ref{sec-proof-t2}
are devoted to the proofs of Theorems \ref{solitonmain}
and \ref{secondthm}, respectively.
These theorems  are also used to prove the existence of a solution to equation 
\eqref{dr_soliton_gas2} in the Subsection \ref{sect-sol-dr2}. 
Properties of the minimizer $\mu^*$ of of $J_\sigma$, such as its support and smoothness under various 
additional assumptions on $\G^+$ and $\sigma$ are discussed in Section \ref{sect-prop-mu}. 
For example, we show there that if $G\mu^*=\varphi$ on the boundary $\part \Omega\subset \G^+$ of some region $\O$
then $\supp\mu^*\cap\Omega=\emptyset$, see Lemma \ref{lem-s=0}. In particular, in the case of the 
soliton condensate ($\sigma \equiv 0$ on $\G^+$), the support
$\supp \mu^*\subset \partial \G^+$
if $\Gamma^+$ is a 2D compact region.
Assume now that $\G^+$ is a finite collection
of piece-wise $C^\infty$ smooth curves for soliton condensate \eqref{dr_soliton_gas1}. 
Then, according to Corollary \ref{cor-smooth1D},
the solution $u$ is $C^\infty$ smooth on $\G^+$ except for small neighborhoods of the points
of non smoothness of $\G^+$ (which include all the endpoints of $\G^+$).  In Lemmas \ref{lem-sig-1} and \ref{lem-sig-0} 
we describe smoothness of $u$ in the general soliton gas with $\s\geq 0$. Finally, in Section \ref{sec-quasimom-vert-cond} we
prove that in the case of a bound state condensate the solution $u(z)$ to 
\eqref{dr_soliton_gas1} is proportional to the density of the quasimomentum meromorphic differential
for the hyperelliptic Riemann surface $\Rscr$ defined by $\G^+\cup\overline{\G^+}$
respectively, see Theorem \ref{lem-boun-cond}. We then discuss extension of these results to 
the KdV soliton gas. 

\section{Background}\label{sect-backg}

Generally speaking, 
solitons and breathers are localized solutions of integrable systems that can be viewed as “particles” of
complex statistical objects called soliton and breather gases.  The nontrivial relation between the  integrability and randomness
in these gases falls within the framework of ``integrable turbulence'',  introduced by V. Zakharov  in \cite{Za09}.
The latter was  motivated by  the complexity of many  nonlinear wave phenomena in physical systems that can be   modeled by integrable equations. 
In view of the growing evidence of  wide spread presence of the “integrable” gases 
in fluids and nonlinear optical media, see \cite{ElTovbis} and references therein, they present 
a fundamental interest for nonlinear science. 

\subsection{ fNLS  soliton  gas}\label{sec-fNLS-gas}

Let us now briefly describe the relation between the spectral theory for fNLS soliton gas  and equations (NDR)
 \eqref{dr_soliton_gas1}-\eqref{dr_soliton_gas2}. We then show that the NDR for the Kortweg - de Vries
 (KdV) equations are closely related with  \eqref{dr_soliton_gas1}-\eqref{dr_soliton_gas2} when $\G^+\subset i\R^+$. 
The fNLS 
%focusing Nonlinear Schr\"odinger  equation 
has the form
\begin{equation}  \label{NLS}
i  \psi_t +  \psi_{xx} +2 |\psi|^2 \psi=0,
\end{equation}
where $x,t\in \R$ are the space-time variables and  $\psi:\R^2 \to \C$ is the unknown function. 
The  simplest solution of  equation \eqref{NLS} 
is a plane wave 
 \begin{equation}\label{pw}
\psi= q e^{2iq^2  t}, 
%\equiv \psi_0.
\end{equation}
where $q>0$ is the amplitude of the wave.

It is well known that the fNLS is an integrable equation \cite{ZS}; the Cauchy (initial value) problem for \eqref{NLS}  can
be solved using the inverse scattering transform (IST) method for different classes of initial data, also known as potentials.
The scattering transform connects a given potential with scattering data expressed in terms of the spectral variable $z\in\C$.
In particular, the scattering data consisting of one pair of spectral points $z=a\pm ib$, where $b>0$, and a (norming) constant $c\in \C$,
defines the famous  soliton solution 
\begin{equation}\label{nls_soliton}
\psi_S (x,t)= 2ib \, \hbox{sech}[2b(x+4at-x_0)]e^{-2i(ax + 2(a^2-b^2)t)+i\phi_0},
   \end{equation}
to the fNLS, where $c$ defines its initial position $x_0$ and the initial phase $\phi_0$.
The soliton \eqref{nls_soliton} 
represents a spatially localized traveling wave  
 (pulse)  on a zero background. 
It is characterized by two independent papameters:  $b=\Im z$ determines the soliton amplitude $2b$ and
$a=\Re z$ determines its velocity $s=-4 a$. A scattering data that consists of several points
$z_j \in \C^+ $ (and their complex conjugates), $j\in \N$, together with their norming constants corresponds to the multi-soliton
solutions. Assuming that originally (at $t=0$) the centers of individual solitons are far from each other, 
we can represent the fNLS time evolution of a multi-soliton solution as propagation and interaction of the 
individual solitons.
It is well known 
that the interaction of solitons in multi-soliton fNLS solutions reduces to only two-soliton elastic
collisions, where
the faster soliton (corresponding to $z_m$) gets a forward shift \cite{ZS}
$$
\Delta_{mj}=\frac1{\Im(z_m)}\log\left|\frac{z_m-\overline{z}_j}{z_m-z_j}
\right|,\quad \Re(z_m)>\Re(z_j),
$$
and the slower ``$z_j$-soliton'' is shifted backwards by $-\Delta_{mj}$.

Suppose now we have a ``gas'' of solitons \eqref{nls_soliton} whose  spectral characteristics $z$ are 
distributed over a compact $\G^+\subset \C^+$ according to some non negative measure $\m$. Assume also that 
the centers of these solitons are distributed uniformly on $\R$ and that $\mu(\G^+)$ is small, i.e the gas 
is diluted. Let us consider the speed of the trial $z$ - soliton in the gas. Since it undergoes rare but
sustained collisions with other solitons, the speed $s_0(z)=-4\Re z$ of a free solution must be modified as
\begin{equation} \label{mod-speed}
s(z)=s_0(z) + \frac{1}{\Im z} \int_{\G^+}\log \le|\frac{w-\bar z}{w-z}\ri|[s_0(z)-s_0(w)]d\mu(w).
\end{equation}
Similar modified speed formula was first obtain by V. Zakharov \cite{Zakharov} in the context of the KdV
equation. How  can one find $s(z)$   without the assumption of the diluted gas, that is, 
when $\mu(\G^+)=O(1)$? The answer is given by the integral equation
\begin{equation} \label{eq-state}
s(z)=s_0(z) + \frac{1}{\Im z} \int_{\Gamma^+}\log \left|\frac{w-\bar z}{w-z}\right|[s(z)-s(w)]d\mu(w)
\end{equation}
for $s(z)$, known as equation of state for the soliton gas, which was first obtained in \cite{ElKamch} using 
purely physical reasoning. Similar equation in the KdV context was obtained in \cite{El2003}. 
In this equation $s(z)$ has the meaning of the speed of the ``element of the gas'' associated with the spectral
parameter $z$ (note that when $\mu(\G^+)=O(1)$ we cannot distinguish individual solitons).
If we now assume a (weak)  dependence of  $s$ and $u$ on space time parameters $x,t$  
(where $d\mu=u d\lambda$  with $\lambda$ being the Lebesgue measure) that  
occurs on much larger spatiotemporal scales than the typical scales
of the oscillations corresponding to individual solitons,  then we complement the equation of state 
\eqref{eq-state} by the continuity equation for the density of states
\begin{equation} \label{kinet}
\part_t u+\part_x (su)=0, 
\end{equation}
which was first suggested in  \cite{ElKamch} and derived in \cite{ElTovbis}. Equations \eqref{eq-state},
\eqref{kinet} form the kinetic equation for dynamic (non-equilibrium) fNLS soliton gas. 
The kinetic equation for the KdV soliton gas was derived in  \cite{El2003}.
It is remarkable that 
recently the kinetic equation having
similar structure 
was derived in the framework of the “generalized hydrodynamics”
for quantum many-body integrable systems, see, for example, \cite{DYC, DSY, VY}. 

It is interesting to observe that \eqref{eq-state} is a direct consequence of
\eqref{dr_soliton_gas1}-\eqref{dr_soliton_gas2}, where $s(z)= \frac{v(z)}{u(z)}$.
%as speed of the ``component'' of the soliton gas parametrized by $z\in \G^+$. 
Indeed, multiplying  \eqref{dr_soliton_gas1} by $s(z)$,
substituting $v(z)=s(z)u(z)$ into  \eqref{dr_soliton_gas2}, subtracting the second equation from the first one
and dividing both parts by $\Im z$ we obtain exactly \eqref{eq-state}.

A mathematical albeit formal (i.e., without error estimates)  derivation of the equation of state \eqref{eq-state} 
was presented in the recent
paper \cite{ElTovbis}. The first step is derivation of equations \eqref{dr_soliton_gas1}-\eqref{dr_soliton_gas2}, which describe
the density of states $u$ and its temporal analog $v$.
The derivation  is based on the idea of  thermodynamic limit for a family of
finite gap solutions of the fNLS, which was originally developed for the KdV equation in \cite{El2003}. 
Finite-gap solutions are quasi-periodic functions in $x,t$ that spectrally can be represented by a finite
number of Schwarz symmetrical arcs (bands) on the complex $z$ plane. Here Schwarz symmetry means that either
a band $\g$ coincides with its Schwarz symmetrical image $\bar \g$ or if $\g$ is a band then  $\bar\g$
is another band. Assume additionally that there is a complex 
constant (initial phase) associated with each band that also respects the Schwarz symmetry, i.e.,  Schwarz symmetrical bands 
have Schwarz symmetrical phases. Given a finite set of Schwarz symmetrical bands with the corresponding phases,
a finite-gap solution to the fNLS can be written explicitly in terms of the Riemann Theta functions
on the hyperelliptic Riemann surface $\mathfrak R$, where  the bands are the branchcuts of $\mathfrak R$, see, for example, \cite{Belokolos}.

For convenience of the further exposition, we will consider $\mathfrak R$ of the genus $2N$, where the genus of $\mathfrak R$ 
is the number of bands minus one. The one exceptional band $\g_0$ will be crossing $\R$, whereas the remaining $N$
bands $\g_j\subset \C^+$, $j=1,\dots,N$, and their Schwarz symmetrical $\g_{-j}:=\bar\g_j\subset \C^-$. 
It was shown in \cite{ElTovbis} that 
the wavenumbers $k_j$, $\tilde k_j$ and the frequencies $\o_j$, $\tilde \o_j$ of a quasi-periodic
finite gap solution $\psi_{2N}$ determined by $\mathfrak R$ can be expressed as
\begin{eqnarray}
k_j &=& -\oint_{\mathrm{A}_j}dp, \quad \o_j = -\oint_{\mathrm{A}_j}dq, \quad j=1, \dots, N, \label{kdp}\\
\tilde k_j &=& \oint_{\mathrm{B}_j}dp, \quad \tilde \o_j = \oint_{\mathrm{B}_j}dq, \quad j = 1, \dots, N, \label{omdq}
\end{eqnarray}
where the cycles $A_j,B_j$ are shown on Figure \ref{Fig:Cont}. Here
$dp(z)$ and $dq(z)$, known as the quasimomentum and quasienergy differentials,  are 
meromorphic differentials on $\mathfrak{R}$ with the only poles at $z=\infty$ on both sheets. These differentials 
are real normalized (all the periods of $dp,dq$ are real) and 
are  (uniquely) 
defined 
(see e.g. \cite{ForLee}, \cite{bertolatovbis2015})
by local expansions 
\begin{equation}\label{pq}
dp \sim \pm 1 + \mathcal{O}(z^{-2}) , \qquad dq \sim \pm 4z+\mathcal{O}(z^{-2}) 
\end{equation}
near $z=\infty$ on the main and second sheet respectively.

\begin{figure}
\centering
\includegraphics[width= 0.5 \linewidth]{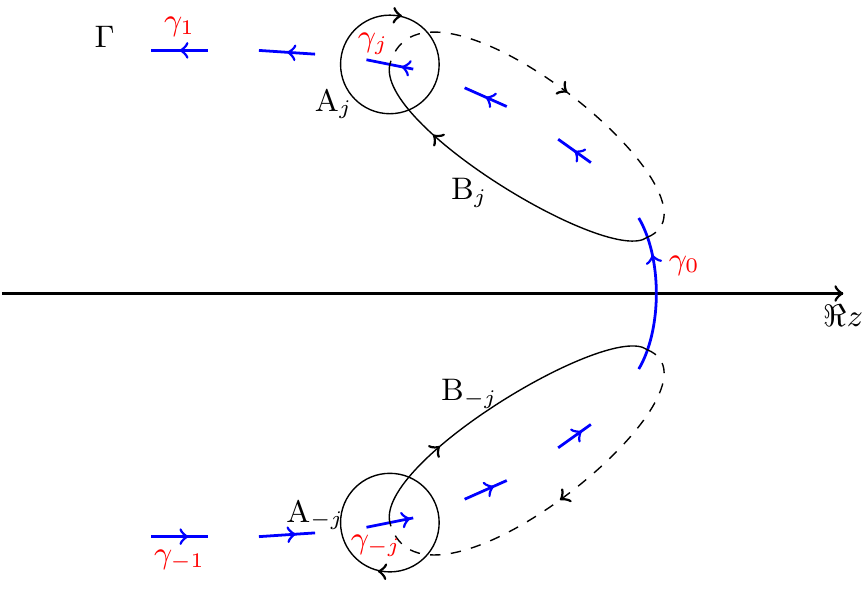}
 \caption{The spectral bands $\g_{\pm j}$ and the cycles $\mathrm{A}_{\pm j}, \mathrm{B}_{\pm j}.$
 The 1D Schwarz symmetrical curve $\G$ consists of the bands $\g_{\pm j}$, $j=0,\dots,N$, and gaps between the bands (the gaps are 
 not shown on this figure).}
 \label{Fig:Cont}
\end{figure}

We shall call the special set of wavenumbers and frequencies defined by  \eqref{kdp}, \eqref{omdq} {\it the  fundamental wavenumber-frequency set}. 
 We note that the wavenumbers and frequencies defined by  \eqref{kdp} and those defined by \eqref{omdq} are of essentially different nature:   in the limit of $\g_j$  shrinking to a point,  we have
 \begin{equation}
  k_j, \o_j \to 0,  \quad \tilde k_j, \tilde \o_j = \mathcal{O}(1),  \qquad j=1, \dots, N,\label{breather_lim} 
  \end{equation}
see  \cite{ElTovbis}.  Motivated by these properties, $k_j,\o_j$  are called
{\it solitonic  wavenumbers and  frequencies } whereas 
the remaining 
 $\tilde k_j$, $\tilde \o_j$  are called {\it carrier  wavenumbers and  frequencies}.
 
The standard normalized holomorphic differentials $w_j$ of  $\mathfrak{R}$ are defined by
\begin{equation} \label{D-P}
w_j=[P_j(z)/R(z)] dz, \quad \oint_{\mathrm{A_i}}w_j = \delta_{ij}, \quad i, j= \pm 1, \dots, \pm N,
\end{equation}
where the polynomials 
\begin{equation} \label{Pj}
P_j(z)=\k_{j,1}z^{2N-1}+\k_{j,2}z^{2M-2}+ \dots +\k_{j,2N}
\end{equation}
 have complex coefficients and the radical
\begin{equation}
\label{rsurf}
R(z)=\prod_{k=1}^{2N+1}(z-\a_k)^\hf(z-\bar\a_k)^\hf.
\end{equation}
defines the hyperelliptic surface  $\mathfrak{R}$, i.e, the product runs over all the endpoints (of the bands) $\a_k\in\C^+$.
It was shown in \cite{ElTovbis} that the  solitonic wavenumbers  and frequencies satisfy the  systems  
\begin{align} 
\sum_{|m|=1}^{N} k_m\Im\oint_{\mathrm{B}_m}{{P_j(\z)d\z}\over{R(\z)}} &={4}\pi \Re\k_{j,1},\nonumber \\
\sum_{|m|=1}^{N} \o_m\Im\oint_{\mathrm{B}_m}{{P_j(\z)d\z} 
\over{R(\z)}} & ={8}\pi \Re\left(\k_{j,1}\sum_{k=1}^{2N+1}\Re\a_k +\k_{j,2}\right), \nonumber \\
& \qquad  |j|= 1,\dots, N, \label{WFR}
 \end{align}
where the latter summation is taken over all the endpoints in $\C^+$. 
We call  \eqref{WFR} the solitonic nonlinear dispersion relations (NDR).
Indeed, the NDR
 indirectly connect (through the Riemann surface $\mathfrak{R}$) 
the solitonic wavenumbers and frequencies of the finite gap solution  $\psi_{2N}$,
i.e., \eqref{WFR} represents  nonlinear dispersion relations. 

Equations \eqref{WFR} together with \eqref{breather_lim} is our starting point  for deriving 
equations \eqref{dr_soliton_gas1}-\eqref{dr_soliton_gas2}. Before describing the derivation,
we want to point out that the matrix of the systems  \eqref{WFR} is negative-definite and, 
therefore, each of the systems
\eqref{dr_soliton_gas1}-\eqref{dr_soliton_gas2} has a unique solution.
The negative-definiteness of the matrix of the systems   \eqref{dr_soliton_gas1}-\eqref{dr_soliton_gas2}
follows from the properties of the the Riemann period matrix $\t$ of the Riemann surface $\mathfrak{R}$
( $\Im \t$ is positive definite).

Suppose now that we start shrinking each band to a point. Then we will be taking the finite gap solution
to its multi-soliton solution  limit, where the phases should be transformed into the corresponding norming constants.
The idea of thermodynamical limit consists of increasing the number $2N+1$ of bands  simultaneously with shrinking the size $2\d_j$
of each  band $\g_j$ (except, possibly, $\g_0$) at  some exponential rate with respect to $N$, so that the centers $z_j$ of the bands located in $\C^+$ 
will be filling a certain compact $\G^+\subset \C^+$ with some limiting density $\phi(z)$. 
Moreover, we assume  the distance between any bands to be 
much larger than the size of the bands. Under these assumptions one can show that
 the leading order behavior of the coefficients of the linear system \eqref{WFR} is given by
\begin{equation} \label{ass-coeff}
\oint_{\tilde{\mathrm{\bf B}}_m}{{P_j(\z)d\z}\over{R(\z)}}=\frac{1}{i\pi}
\le[ \log\frac{R_0(z_j)R_0(z_m)+z_jz_m+\d_0^2}
{R_0(z_j)R_0(\bar z_m)+z_j\bar z_m+\d_0^2}
- \log \frac{z_m-z_j}{z_m-\bar z_j}\ri] 
\end{equation}
when $m\neq j$ and 
\[ \oint_{\tilde{\mathrm{ B}}_j}{{P_j(\z)d\z}\over{R(\z)}}=i\frac{2\log\d_j}{\pi},
\]
where $\tilde{\mathrm{B}}_m =\mathrm{ B}_m + \mathrm{ B}_{-m}$,
$2\d_0$
%=\a_1-\bar\a_1$ 
is the distance between the endpoints of the exceptional arc $\g_0$ 
and
$R_0(z):=\sqrt{z^2+\d_0^2}\sim z$ as $z\to \infty$ (here WLOG we assume $z_0=0$).
 The imaginary part of  \eqref{ass-coeff} provides 
the expression for the 
 kernel of the integral equations for $u(z)$ in the case of the fNLS breather gas 
 \begin{multline} \label{dr_breather_gas1}
\frac 1{\pi} 
\int_{\G^+}
\le[\log\le| \frac{w-\bar z}{w-z}\ri|+ \log\le|\frac{R_0(z)R_0(w)+z w +\d_0^2}
{R_0(\bar z)R_0(w)+\bar z w+\d_0^2}\ri|\ri]   u(w) d\lambda(w)
+\sigma(z)u(z) \\
= \Im R_0(z), 
\end{multline}
whereas the second formula gives rise to
 the secular term $\s(z) u(z)$ in \eqref{dr_breather_gas1},
 where $\s(z)$ is defined by 
$\phi(z)$ and by the rate the bands shrink  near $z$.
  The same holds for the integral equation
 for $v(z)$:
\begin{multline}  \label{dr_breather_gas2}
\frac 1{\pi} \int_{\Gamma^+}
\left[\log\left| \frac{w-\bar z}{w-z}\right|+\log\left|\frac{R_0(z)R_0(w)+ z w+\delta_0^2}
{R_0(\bar z)R_0(w)+ \bar z w+\delta_0^2}\right|\right] v(w) d\lambda(w)
+\sigma(z)v(z) \\
= - 2\Im[z R_0(z)].
\end{multline}
 
The breather gas is obtained when in the thermodynamic limit all the bands except $\gamma_0$ are shrinking to points while
the exceptional band $\gamma_0$ approaches some limiting position as $N \to \infty$, where the endpoints of $\gamma_0$
approach $\pm i\delta_0$ respectively. Being considered alone, the limiting spectral band $\gamma_0$ corresponds
to the plane wave solution \eqref{pw} with $q=\delta_0$. The band $\gamma_0$ together with  Schwarz symmetrical points
of discrete spectrum $z, \bar z$
correspond to a soliton on the plane wave (carrier) background, also known as 
a breather. It is  remarkable that the kernel in the integral equations \eqref{dr_breather_gas1}-\eqref{dr_breather_gas2}, 
being divided by $\Im R_0(z)$, provides an elegant expression for the ``position shift'' of two interacting 
breathers; some considerably more  involved expressions for this phase shift were recently obtained in 
\cite{Gino1, Gino2, Gelash}.
Therefore, equations 
 \eqref{dr_breather_gas1}-\eqref{dr_breather_gas2} 
represent nonlinear dispersive relations for the breather gas.
It is easy to check that equations   \eqref{dr_breather_gas1}-\eqref{dr_breather_gas2}
coincide with  \eqref{dr_soliton_gas1}-\eqref{dr_soliton_gas2}  in  the limit $\delta_0\to 0$.
Thus, soliton gas can be considered as a particular case of the breather gas, see \cite{ElTovbis} for details. 

In the case of subexponential  rate of shrinking of bands $\g_j$ in the thermodynamic limit, the function 
$\s(z)$ turns to be zero and we obtain a breather (or soliton, if $\d_0\ra 0$) condensate (\cite{ElTovbis}).
As it was mentioned in Remark \ref{prop-cond}, the term ``condensate'' reflects the fact that for a 
given $\G^+$ and $\varphi(z)=\Im z$ the energy $J_\s(\m^*_\s)$ is minimized when $\s\equiv 0$ on $\G^+$.

Consider a sequence of atomic, possibly signed
measures $\m_N$ with weights 
\begin{equation}
u_j=\frac{\phi(z_j) k_j}{2\pi}
\end{equation}
at each $z_j$, $j=1,\dots,N$. 
Assuming that the sequence $\{\m_N\}_1^\infty$ weakly converges  to some
measure $d\mu = ud\lambda$ on $\mathbb C^+$, we obtain integral  equation  \eqref{dr_soliton_gas1}
as the
thermodynamic limit of the first equation \eqref{WFR} (here $\gamma_0$ also shrinks to a point $z_0=0$). 
Equation \eqref{dr_soliton_gas2} can be obtained from the second  \eqref{WFR} equation similarly.

We want to emphasize that the existence and uniqueness of solution of  \eqref{WFR} do not  imply the 
existence and uniqueness of solutions  \eqref{dr_soliton_gas1}-\eqref{dr_soliton_gas2} and, 
what is especially important, provides no information related to the requirement $u(z)\geq 0$ on $\Gamma^+$.
The aim of the present paper is to address these questions.

\subsection{  Bound state fNLS  and KdV gases}\label{sec-quasimom-vert}

Soliton gas is called a bound state gas if $\Gamma^+$ is a subset of a vertical line $\Re z =c$, where $c$ is a constant.
The terminology comes from the fact that all the solitons in a multisoliton fNLS solution with the 
(discrete) spectrum on  $\Re z =c$ have the same speed and therefore such a solution does not decompose
into a collection of individual solitons in the process of evolution. In the context of soliton gases, one can note
that solutions to the NDR  \eqref{dr_soliton_gas1}-\eqref{dr_soliton_gas2} for bound state gases are proportional
since these equations have proportional right 
hand sides. Thus, all components of a bound state gas have the same speed $-4c$.  

According to \cite{ElTovbis},  equations \eqref{dr_soliton_gas1}-\eqref{dr_soliton_gas2} with $\G^+=[0,iq]$ 
and $\s\equiv 0$ have solutions 
\begin{equation} \label{box-cond}
u(z)=\frac{-iz}{\pi\sqrt{z^2+q^2}}, \quad v\equiv 0, \qquad \text{ on $[0,iq]$},                     
\end{equation}
which, according to Corollary \ref{cor-smooth1D} in Section \ref{sect-prop-mu} are $C^\infty$ smooth (in fact, analytic) on any proper subarc of 
$\G^+$. The only singularity $z_*=iq$ is at the upper endpoint of $\G^+$, which is a point of non-smoothness of $\G^+$,
see Remark \ref{rem-endp-t}, Section \ref{sect-prop-mu}.
We note that the local behavior of $u$ near $z_*$ is in full agreement with Remark \ref{rem-endp-t}. 

\begin{remark} \label{rem-int-real}
The reader may notice that in the example above $\Gamma^+\cap\R=\{0\}\neq\emptyset$, so, as stated, Theorems \ref{solitonmain}
and \ref{secondthm} are not applicable to this $\Gamma^+$. However, our results, not included in this paper, show
that these theorems are still applicable to the case when a  1D curve $\Gamma^+$ intersects $\R$ transversally.
\end{remark}

Solution \eqref{box-cond} of $Gu(z)=\Im z$ was obtained  by first extending \eqref{dr_soliton_gas1} symmetrically
to $\C^-$ (see equation 
\eqref{full-eq1}, Section \ref{sect-prop-mu}), then differentiating both 
sides in $s=\Im z$
and, finally,  inverting the obtained Finite Hilbert Transform (FHT) on $[-iq,iq]$.  We will use this approach 
in Section \ref{sec-quasimom-vert-cond} below to 
solve \eqref{dr_soliton_gas1} for any bound state condensate. 

It is interesting to observe that the NDR for the KdV soliton gas 
\begin{align} \label{dr_kdv-soliton_gas1}
 \frac 1{\pi} \int _{\Gamma^+}\log \left|\frac{\o+\z}{\o-\z}\right|
u(\o)d\o+\sigma(\z)u(\z)& = \frac\z 2,\\
\label{dr_kdv-soliton_gas2}
 \frac 1{\pi} \int _{\Gamma^+}\log \left|\frac{\o+\z}{\o-\z}\right|
v(\o)d\o+\sigma(\z)v(\z)& = -2\z^3,
\end{align}
first obtained in \cite{El2003}, are closely related with the fNLS bound state NDR  with $\G^+\subset i\R^+$. Indeed,
substituting $z=i\z$ and $w=i\o$, we convert
the left hand sides of  \eqref{dr_soliton_gas1}-\eqref{dr_soliton_gas2} into the left hand sides of 
\eqref{dr_kdv-soliton_gas1}-\eqref{dr_kdv-soliton_gas2}. In fact, solutions  of the corresponding
first NDRs coincide up to the factor 2. Therefore, all the obtained in this paper results applicable to  equation
\eqref{dr_soliton_gas1} and most applicable to equation \eqref{dr_soliton_gas2} with $\G^+\subset i\R^+$ are automatically applicable
to the KdV soliton gas from \cite{El2003}, see Section \ref{sec-quasimom-vert-cond} for details.

\begin{remark}
Realizations of an fNLS soliton  gas  can be related with  the semiclassical limit of the fNLS equation 
with rapidly decaying real one hump potential that typically has $O(1/\e)$ points of discrete spectrum
(solitons) located on $i\R^+$, where $\e>0$  is a small semiclassical  parameter. 
Such potentials include, for example, $\sech\, x$, the barrier (box) potential and many others.
The fNLS time evolution of such potentials is known to typically
lead to the appearance of coherent structures of increasing
complexity that can be locally approximated by genus $n$
finite-gap solutions with $n$ increasing in time \cite{EKT}.
There are strong indications
 that for sufficiently large time $t$ (and consequently large
$n$), the semiclassical spectrum of these solutions fits into
one of the thermodynamic scaling requirements described above.
Taking into account the effective randomization
of phases, the large $t$ evolution of semiclassical solutions
is expected to provide the dynamical realization of a bound state soliton
gas studied in this paper. It is interesting that the first rigorous study  of the 
large $n$ limit of a special $n$-soliton solution
to the KdV  was recently conducted in \cite{Girotti}. It is based on the idea of the primitive
potential from \cite{DZZ}.
\end{remark}

\section{Proof of Theorem \ref{solitonmain}} \label{proofmain}

\subsection{Proof of part (a): the variational problem}\label{sec-var-prob}

We are going to show that the energy functional $J_{\sigma}$ 
defined in 
\eqref{def:Jsigma} 
has a unique minimizer on the set of Borel measures on $\Gamma^+$.
As a first step we show that $J_{\sigma}$ is lower semicontinuous.

\begin{lemma} \label{Jlsc}
	The functional $J_{\sigma}$ is lower semicontinuous
	on the set of positive Borel measures on $\Gamma^+$ with the weak$^*$
	topology.
\end{lemma}
\begin{proof}
	Let $(\mu_k)_k$ be a sequence of positive Borel measures
	on $\Gamma^+$ with $\mu$ as the weak$^*$ limit.
	We have to show that
	\begin{align} \label{Jlscproof1} 
	J_{\sigma}(\mu) \leq \liminf_k J_{\sigma}(\mu_k). 
	\end{align}
	To do so, we may assume (by passing to a subsequence
	if necessary) that $J_{\sigma}(\mu_k) < +\infty$ for
	every $k$, and that the limit
	\begin{align}  \label{Jlscproof2}
	J^* := \lim_{k \to \infty} J_{\sigma}(\mu_k) 
	\end{align}
	exists with $J^* < + \infty$. (If it would be infinite,
	there would be nothing to prove).
	
	It is known that the quadratic term $\mu \mapsto \int G(\mu) d\mu$
	in \eqref{Jmu}
	is lower semicontinuous \cite{SaffTotik}, while the linear term
	$\mu \mapsto \int \varphi d\mu$ is continuous with
	respect to the weak$^*$ topology. Thus to prove
	\eqref{Jlscproof1} it suffices to show that 
	$\sigma \mu$ is absolutely continuous with respect
	to $\lambda$, say $\sigma \mu = \sigma u \lambda$, and
	\begin{align} \label{Jlscproof3}
	\int \sigma u^2 d\lambda \leq \liminf_{k \to \infty} \int \sigma u_k^2 d\lambda,
	\end{align}
	and this is what we are going to do.
	
	For each $k$ we have that $J_{\sigma}(\mu_k)$ is finite 
	and thus by the definition \eqref{Jmu} there exists
	a non-negative measurable function $u_k$ on $\Gamma^+$
	such that $\sigma \mu_k = \sigma u_k \lambda$.
	Since $\int G\mu_k d\mu_k \geq 0$, we have
	\begin{align*} \int \sigma u_k^2 d\lambda
	& = J_{\sigma}(\mu_k) - J_0(\mu_k)  \\
	& \leq J_{\sigma}(\mu_k) + 2 \int \varphi d\mu_k 
	\to J^* + 2 \int \varphi d\mu 
	\quad \text{ as } k \to \infty, 
	\end{align*} 
	by weak$^*$ convergence. Hence the integrals
	$\int \sigma u_k^2 d\lambda$ remain bounded as $k \to \infty$, 
	and passing
	to a further subsequence, if necessary, we may assume that
	\begin{align} \label{Jlscproof4}
	\lim_{k \to \infty}  \int \sigma u_k^2 d\lambda = R^2 
	\end{align}
	exists and $R$ is finite.
	
	Let $\psi$ be a continuous function on $\Gamma^+$.
	Then by the Cauchy-Schwarz inequality
	\begin{align*} 
	\int_{\Gamma^+} |\psi| \sigma^{1/2} d\mu_k
	= \int_{\Gamma^+} |\psi| \sigma^{1/2} u_k d\lambda
	& \leq
	\left( \int_{\Gamma^+} |\psi|^2 d\lambda \right)^{1/2}
	\left( \int_{\Gamma^+} \sigma u_k^2 d\lambda \right)^{1/2}. 
	\end{align*}	
	Taking the limit $k \to \infty$ we obtain by 
	weak$^*$ convergence $\mu_k \to \mu$ and \eqref{Jlscproof4} that 
	\begin{align} \label{Jlscproof5}  
	\int |\psi| \sigma^{1/2} d\mu \leq
	R \left( \int_{\Gamma^+} |\psi|^2 d\lambda \right)^{1/2} 
	\end{align}
	for any continuous function $\psi$ on $\Gamma^+$. 
	Since continuous functions are dense in $L^2(\Gamma^+, \lambda)$
	the inequality \eqref{Jlscproof5} continues to hold
	for every $\psi \in L^2(\Gamma^+, \lambda)$.
	
	We take $\psi = 1_A$, where $A$ is a Borel subset of
	$\Gamma^+$ and $1_A$ denotes the characteristic function of $A$. Then \eqref{Jlscproof5} gives
	\[ \int_A \sigma^{1/2} d\mu \leq R \sqrt{\lambda(A)}, \]
	which implies that $\sigma^{1/2} \mu$ is absolutely
	continuous with respect to $\lambda$. Hence there is a non-negative
	density $u$ such that $\sigma^{1/2} \mu = \sigma^{1/2} u \lambda$
	and then also
	\begin{align} \label{Jlscproof6} 
	\sigma \mu =  \sigma u \lambda. 
	\end{align}
	
	Next take $\psi_M = \min(\sigma^{1/2} u, M)$ for some $M > 0$.
	Then $\psi_M$ is bounded and thus certainly in $L^2(\Gamma^+, \lambda)$.
	Hence by \eqref{Jlscproof5} and \eqref{Jlscproof6}
	\[ \int_{\Gamma^+} \psi_M \sigma^{1/2} u d\lambda 
	\leq R \left( \int_{\Gamma^+} \psi_M^2  d\lambda  \right)^{1/2}
	\leq R \left(\int_{\Gamma^+} \psi_M \sigma^{1/2} u d\lambda\right)^{1/2} \]
	since $\psi_M \leq \sigma^{1/2} u$.
	Since the integral is finite we deduce
	\[ \int_{\Gamma^+} \psi_M \sigma^{1/2} u d\lambda \leq R^2. \]
	Lettting $M \to +\infty$ and noting that $\psi_M \nearrow
	\sigma^{1/2} u$, we obtain by monotone convergence
	and \eqref{Jlscproof4}
	\begin{align} \label{Jlscproof7}
	\int \sigma u^2 d\lambda \leq R^2 = \lim_{k \to \infty}
	\int \sigma u_k^2 d\lambda. 
	\end{align}
	This implies \eqref{Jlscproof3} and the lemma is proven.
\end{proof}

\begin{lemma} \label{Jsublevel}
	Suppose $\Gamma^+ \cap \mathbb R = \emptyset$.
	Then 
	$J_{\sigma}$ has compact sub-level sets, i.e.,
	for every $c \in \mathbb R$ the set 
	\begin{align} \label{Jsublevel1} 
	\{ \mu \geq 0 \mid J_{\sigma}(\mu) \leq c \} 
	\end{align}
	is compact in the weak$^*$ topology.
\end{lemma}
\begin{proof}	
	The functional $\mu \mapsto \int G(\mu) d\mu$ has a unique minimizer among 
	probability measures on $\Gamma^+$, and
	the minimum is positive, say $c_0 > 0$,
	since $\Gamma^+ \cap \mathbb R = \emptyset$. 
	Then, as it is a quadratic functional, 
	\[ \int G(\mu) d \mu  \geq c_0  \left(\int d\mu \right)^2 \]
	By Definition \ref{def:Jsigma}, since $\sigma \geq 0$,
	\begin{equation} \label{Jsublevel2} 
	J_{\sigma}(\mu) \geq
	J_0(\mu)  \geq  c_0 \left( \int d\mu\right)^2
	- 2\max_{\Gamma^+} \varphi \int d\mu. 
	\end{equation}		
	This immediately implies that for every $c \in \mathbb R$
	there is $M > 0$ such that for every $\mu \geq 0$ on  $\Gamma^+$ we have
	\[ J_{\sigma}(\mu) \leq c \implies \int d\mu \leq M. \]
	In other words, the sub-level set \eqref{Jsublevel1}
	is contained in the set
	\begin{equation} \label{Jsublevel3} 
	\{ \mu \geq 0 \mid  \int d \mu \leq M \} \end{equation}
	which is compact in the weak$^*$ topology since $\Gamma^+$ is
	compact. Because of the lower semi-contuinity of $J_{\sigma}$,
	see Lemma \ref{Jlsc}, the set \eqref{Jsublevel1} is
	a closed subset of \eqref{Jsublevel3} and the lemma follows.
\end{proof}

Now we can prove part (a).
\begin{proof}[Proof of Theorem \ref{solitonmain} (a)]
	It follows from \eqref{Jsublevel2} that $J_{\sigma}$
	is bounded away from $-\infty$. Let $(\mu_k)_k$ be
	a sequence of non-negative measures on $\Gamma^+$ such	that
	\begin{align} \label{Jmin1} 
	\lim_{k \to \infty} J_{\sigma}(\mu_k)
	= \inf_{\mu \geq 0} J_{\sigma}(\mu). 
	\end{align}
	Because of Lemma \ref{Jsublevel} the sequence $(\mu_k)_k$
	is in a weak$^*$ compact set, and therefore it has
	a subsequence with a weak$^*$ limit, say $\mu^*$.
	Because of Lemma \ref{Jlsc} and \eqref{Jmin1} we then
	have
	\[ J_{\sigma}(\mu^*) \leq \inf_{\mu \geq 0} J_{\sigma}(\mu) \]
	which implies that $\mu^*$ is a minimizer.
	
	The minimizer is unique, since the functional $J_{\sigma}$
	is strictly convex, which follows from 
	Definition \ref{def:Jsigma} since $J_0$ is 
	known to be strictly convex.  	
	From Definition \ref{def:Jsigma} it is also clear
	that $\sigma \mu^*$ is continuous with respect
	to $\lambda$. This completes the proof of part (a).
\end{proof}

\subsection{Proof of part (b): variational condition on $\supp(\mu^*)$}

The proof relies on standard arguments  in variational calculus. Since we use these arguments later on as well,
we state them in a separate lemma.
\begin{lemma} \label{lemma26}
	Let $\nu$ be a measure on $\Gamma^+$.
	\begin{enumerate}
		\item[\rm (a)]
		If $J_{\sigma}(\nu) < \infty$, then
		\begin{equation} \label{Gvareqintegral1}
		\int (G \mu^* + \sigma u^* - \varphi) d\nu  \geq 0.
		\end{equation}
		\item[\rm (b)] If $ \nu \leq \mu^*$, then
		\begin{equation} \label{Gvareqintegral2}
		\int (G \mu^* + \sigma u^* - \varphi) d\nu  = 0.
		\end{equation}
	\end{enumerate}
\end{lemma}
\begin{proof}
(a) Suppose $J_{\sigma}(\nu) < \infty$. Then
$\sigma \nu = \sigma v \lambda$ for some $v$. Note that
\begin{align*} 
\int \sigma (u^*+\varepsilon v)^2 d\lambda - 
\int \sigma (u^*)^2 d\lambda & 
= 2 \varepsilon \int \sigma u^* v d\lambda
+ O(\varepsilon^2) \\
& = 2 \varepsilon \int \sigma u^* d\nu + O(\varepsilon^2)
\end{align*} as $\varepsilon \to 0$, 
and it follows that 
\begin{equation} \label{Gvareqintegral3}
J_{\sigma}(\mu^* + \varepsilon \nu)
- J_{\sigma}(\mu^*) =
2 \varepsilon
\int \left( G\mu^* +  \sigma u^* - \varphi \right) d\nu
+ O(\varepsilon^2).
\end{equation}
Since $\mu^*$ is the minimizer of $J_{\sigma}$,
the left-hand side of \eqref{Gvareqintegral3} is
non-negative for every $\varepsilon > 0$ 
and \eqref{Gvareqintegral1}
follows by letting $\varepsilon \to 0+$. 

\medskip
(b)
If $\nu \leq \mu^*$, then also $J_{\sigma}(\nu) < \infty$, and then the left-hand side of \eqref{Gvareqintegral3} is non-negative for every 
	$\varepsilon \in (-1, \infty)$. Then 
the equality \eqref{Gvareqintegral2} follows
by letting $\varepsilon \to 0$ through negative
values as well.
\end{proof}

\begin{proof}[Proof of part (b)] 
Let $E = \{ z \in \Gamma^+
\mid (G\mu^* +  \sigma u^* - \varphi)(z) < 0 \}$
and  suppose $\mu^*(E) > 0$. Then there is
	$\delta > 0$ such that
	\[ E_{\delta} = 
	\{ z \in \Gamma^+
	\mid (G\mu^* +  \sigma u^* - \varphi)(z) < - \delta  \}
	\] has $\mu^*(E_{\delta}) > 0$. 
	Let $\nu$ be the restriction of $\mu^*$ to $E_{\delta}$.
	We get
	\[ \int (G\mu^* +  \sigma u^* - \varphi) d \nu < - \delta \nu(E_{\delta}) < 0 \]
	which is in contradiction with \eqref{Gvareqintegral2}, since clearly $\nu \leq \mu^*$. 
	We obtain a similar contradiction in case
	$\mu^*( \{ z \in \Gamma^+
	\mid (G\mu^* +  \sigma u^* - \varphi)(z) > 0 \}) > 0$
	and \eqref{maineq-phi}  follows.
\end{proof}

\subsection{Proof of part (c): variational condition outside $\supp(\mu^*)$}

To obtain the equality \eqref{maineq-phi-outside} on 
$\Gamma^+ \setminus \supp(\mu^+)$ 
we need the conditions of part (c)
Theorem \ref{solitonmain}.  That is, 
we require that $\varphi$ is positive,
continuous and superharmonic on $\mathbb C^+$.
We also use Assumption \ref{assumptiondm} on $\lambda$.
It follows from this assumption
that the  Green potential of the restriction of $\lambda$
to any open subset of $\Gamma^+$ is continuous
as well,
as this is a consequence of the following simple 
general fact.
\begin{lemma} \label{lemma36}
	Suppose $\mu$, $\nu$ are measures with
	$\nu \leq \mu$. If $G \mu$ is continuous
	then $G\nu$ is continuous as well.
\end{lemma}
\begin{proof}
	Being a Green potential,
	$G \nu$ is lower semicontinuous. 
	Since $\mu - \nu \geq 0$, also $G(\mu - \nu)$
	is lower semicontinuous. Hence if $G\mu$
	is continuous, then $G\nu = G\mu - G(\mu -\nu)$
	is upper semicontinuous as well, and
	therefore continuous.
\end{proof}

\begin{proof}[Proof of part (c) of Theorem \ref{solitonmain}]
	Suppose $z_0 \in \supp(\mu^*)$ is such that
	$G \mu^*(z_0) > \varphi(z_0)$.
	Since $G\mu^* - \varphi$ is lower semicontinuous,
	we can then find $\varepsilon > 0$ and
	a disk $D_0 = D(z_0,r_0)$ around $z_0$ such that
	$G \mu^* - \varphi \geq \varepsilon$ on $D_0$.
	The restriction of $\mu^*$ to the complement
	of $D_0$ is a measure that is obviously
	bounded by $\mu^*$, and thus by
	\eqref{Gvareqintegral2} we have
	\[ \int_{\mathbb C \setminus D_0}
	(G \mu^* + \sigma u^* - \varphi) d\mu^* = 0. \] 	
	We conclude, since $G\mu^* + \sigma u^* - \varphi \geq G\mu^* - \varphi \geq \varepsilon$ on $D_0$,
	\begin{align*}
	\int (G \mu^* + \sigma u^*- \varphi) d\mu^*
	& = \int_{D_0} (G \mu^* + \sigma u^* - \varphi) d\mu^* \\
	& \geq \varepsilon \mu^*(D_0) > 0.
	\end{align*}
	The last inequality holds since $D_0$ is a disk around $z_0$
	and $z_0 \in \supp (\mu^*)$. On the other hand, 
	\eqref{Gvareqintegral2}
	also holds for $\nu = \mu^*$ itself, and we find a contradiction.
	
	Thus $G\mu^* \leq \varphi$ on $\supp(\mu^*)$
	and  we first  
	extend the inequality to all of $\mathbb C^+$.
	This will follow from the maximum principle
	for subharmonic functions as we now show.
	Note that  
	\[ \limsup_{z \to x \in \mathbb R}
	(G \mu^* - \varphi)(z)  \leq 0, \]
	since $\varphi$ is non-negative and $G\mu^*$ is
	zero on the real line. Similarly
	\begin{equation} \label{Gmulimit} 
	\limsup_{z \to \infty}
	(G \mu^* - \varphi)(z)  \leq 0. \end{equation}
	Due to the assumption that $\varphi$ is superharmonic,
	and due to the fact that $G\mu^*$ is harmonic
	away from the support of $\mu^*$,
	we also have that 
	$G \mu^* - \varphi$ is subharmonic
	on $\mathbb C^+ \setminus \supp(\mu^*)$.
	Then by the maximum principle for subharmonic functions
	(which says that the maximum is attained on the boundary) 	we indeed have 	that 
	\begin{equation} \label{Gmustar1} G \mu^* \leq  \varphi  \text{ on } \mathbb C^+
	\end{equation}
	
	To prove \eqref{maineq-phi-outside} we take $z_0 \in \Gamma^* \setminus \supp(\mu^*)$ and in order to get a contradiction
	we suppose $(G\mu^* - \varphi)(z_0) \neq 0$.
	Because of \eqref{Gmustar1} we then have
	$(G\mu^* - \varphi)(z_0) < 0$.
	
	Since $G\mu^*-\varphi$ is continuous away from
	the support of $\mu^*$, there is $r_0 > 0$
	and $\varepsilon >0$ such that
	$D_0 = D(z_0,r_0) \subset \Gamma^* \setminus
	\supp(\mu^*)$ and $G\mu^* - \varphi < -\varepsilon \quad	
	\text{on } D_0$.
	Let $\nu$ denote the restriction of $\lambda$ to $D_0$.
	Then $\nu(D_0) > 0$ by Assumption \ref{assumptiondm}
	(since $\supp(\lambda) = \Gamma^+)$, and
	\begin{equation} \label{Gmustar4} 
	\int \left(G\mu^* - \varphi \right) d \nu
	< -\varepsilon \nu(D_0) < 0. 
	\end{equation}
	
	Because of Assumption \ref{assumptiondm} and
	its consequence that is noted before the statement
	of Assumption \ref{assumptiondm}, we 
	have that $G\nu$ is continuous and bounded.
	Hence $\int G\nu d\nu$ is finite. Also
	$\nu$ has a density $v$ with respect to $\lambda$
	(which is just the characteristic function of $D_0 \cap
	\Gamma^+$) and $\int \sigma v^2 d\lambda$ is finite as well.
	Hence by \eqref{Jmu} we have 
	$J_{\sigma}(\nu) < \infty$, and then by
	\eqref{Gvareqintegral1} we get	
	\[ \int (G \mu^*+\sigma u^*-\varphi) d\nu \geq 0. \]
	However $\sigma u^* = 0$ on $\supp(\nu)$,
	and we find a contradiction with \eqref{Gmustar4}.
\end{proof}

A slight extension of the proof also yields the following fact about $\supp \mu^*$.
Let $\Omega$ denote the unbounded  connected component of $\mathbb C^+\setminus \Gamma^+$,
so that $\partial \Omega\subset \R \cup\{\infty\}\cup \Gamma^+$.

\begin{proposition} \label{propGvareq2}
	In the conditions of Theorem \ref{solitonmain}, part (c),
	assume that 
	\begin{equation} \label{Gvareqproof3} 
		\liminf_{ z \to \infty} \varphi(z) > 0.
		\end{equation}
	Then $\partial \Omega\cap \Gamma^+$, i.e., the outer boundary of $\Gamma^+$, is contained in $\supp(\mu^*)$.
	
	In particular, if  $\Gamma^+$ has empty interior and
	a connected complement then $\supp(\mu^*) = \Gamma^+$. 
\end{proposition}
\begin{proof}
	As in the proof of Theorem \ref{solitonmain}, part (c),
	we have \eqref{Gmustar1} but now
	\[ \limsup_{z \to \infty}
		\left( G \mu^* - \varphi\right)(z) < 0. \]
	Then by the maximum principle $G \mu^* - \varphi < 0$
	on the unbounded connected component of
	$\mathbb C^+ \setminus \supp(\mu^*)$. But if $z_0\in \G^+\setminus\supp \m^*$,
	then $(G \mu^* - \varphi)(z_0) = 0$ according to \eqref{maineq-phi-outside}.
	Thus,  $\part \O\cap \G^+ \subset \supp \m^*$.
\end{proof}

\section{Proof of Theorem \ref{secondthm} and the second NDR equation}\label{sec-proof-t2}

\subsection{Proof of part (a)}

\begin{proof}
	Take $z_0 \in S$ with $\sigma(z_0) > 0$.
	Since $\sigma$ is continuous there exist a disk 
	$D = D_0(z_0,r_0)$ around $z_0$ and a number 
	$C_0 >0$  such that $
	\sigma(z) \geq C_0 > 0$ for all $z \in D \cap \Gamma^+$.
	Let $\nu$ be the restriction of $\lambda$ to
	$D \cap \Gamma^+$. 
	Then $G\nu$ is continuous by Assumption \ref{assumptiondm}. 
	Let $\mu_0^*$ be the restriction of $\mu^*$ to
	$D \cap \Gamma^+$.
	Then $\mu_0^*$ has the density $u^*$ on 
	$D \cap \Gamma^+$ and  
	\[ \sigma u^* \leq \varphi, \quad \mu_0^*\text{-a.e.} \]
	which is a consequence of \eqref{maineq-phi}.
	
	Since $\sigma \geq C_0$ on $\supp(\mu_0^*)$,
	we find
	\[ u^* \leq \frac{1}{C_0} \max_{z \in \Gamma^+} \varphi(z), 
	\quad \mu_0^*\text{-a.e.}. \]
	This means that there is a constant $C >0$ such
	that $\mu_0^* \leq C \nu$.
	Since $G\nu$ is continuous, it 
	follows from Lemma \ref{lemma36} that $G\mu_0^*$ is continuous.
	
	Then $G\mu^*$ is continuous on $D \cap \Gamma^+$, and in particular at $z_0$, 
	since $G\mu$ is the sum of $G\mu_0^*$ and $G(\mu^*-\mu_0^*)$, and the latter
	is  continuous on $D \cap \Gamma^+$ as $\mu^*-\mu_0^*$ is supported
	away from this open set.
	Since $z_0 \in S$ is arbitrary, we find that $G\mu^*$ is continuous on 	 $S$.
	
	Because of \eqref{maineq-phi} we have
	\[ u^* = \frac{\varphi - G \mu^*}{\sigma} \quad \mu^*\text{-a.e. on } S. \]
	The right-hand side is continuous on $S$, and if we just 
	redefine $u^*$ by the right-hand side, then the new $u^*$ is still a
	valid density for $\mu^*$ and it is continuous. We also find that
	the identity $G\mu^* + \sigma u^* = \varphi$ holds on $S$, which concludes the proof of part (a).
\end{proof}

\subsection{Proof of part (b)}

We need some notions from potential theory
that we briefly summarize, see \cite{ArmitageGardiner, Helms, Ransford,SaffTotik} for fuller accounts.
A set where a superharmonic function is $+\infty$
is called a polar set, otherwise it is non-polar. 
A compact set $A \subset \mathbb C^+$
is non-polar if and only if its capacity
is positive, which means that there is a 
probability measure $\nu$ with $\supp(\nu) \subset A$ 
and $\int G \nu d\nu < \infty$.

The fine topology on $\mathbb C^+$ is the coarsest topology
for which all Green potentials $G\mu$, $\mu \geq 0$
are continuous. The fine topology is finer than
the usual Euclidean topology, since there exist non-continuous Green potentials.
A fine neighborhood of $z_0$ is a neighborhood in
the fine topology.
Then a set $S$ is thick at $z_0$ if and only if
every fine neighborhood of $z_0$ has a non-empty
intersection with $S \setminus \{z_0\}$,
see e.g.\ \cite[Theorem 7.2.3]{ArmitageGardiner}.

\begin{lemma}
	Suppose $S$ is thick at $z_0$. Let $U$ be
	a fine neighborhood of $z_0$. Then $S \cap U$ is non-polar.
\end{lemma}
\begin{proof}
	Suppose $S \cap U$ is polar. By
	\cite[Theorem 7.2.2]{ArmitageGardiner} 
	a polar set is thin everywhere, so in particular
	$S \cap U$ is thin at $z_0$. Thus 
	there is a fine neighborhood  $V$ of $z_0$
	that does not intersect $(S \cap U ) \setminus \{z_0\}$.
	
	Then $U \cap V$ is a fine neighborhood of $z_0$
	that does not intersect $S \setminus \{z_0\}$
	which is a contradiction, since $S$ is thick at $z_0$. 
\end{proof}

Now we turn to the proof of part (b) of Theorem \ref{secondthm}.

\begin{proof}[Proof of part (b).] 
	Suppose $S_0$ is thick at $z_0$, and assume $G\mu^*(z_0) \neq  \varphi(z_0)$. 	
	From Theorem \ref{solitonmain} (c) we know that $G \mu^*(z_0) \leq \varphi(z_0)$, 
	and therefore there is $\delta > 0$ 	such that
	\[ G \mu^*(z_0) < \varphi(z_0) - \delta. \]
	Then
	$A = \{ z \in \mathbb C^+ \mid G \mu^*(z) \leq \varphi(z) - \delta \}$
	is a fine neighborhood of $z_0$, and it is
	also closed in the usual topology,
	since $G\mu^*$ is lower semicontinuous
	and $\varphi$ is continuous.
	
	Since $S_0$ is thick at $z_0$, we conclude that
	$A \cap S_0$ has positive capacity.
	Thus there is a probability measure $\nu$ on $A \cap S_0$ with
	$\int G\nu d\nu < \infty$. Since $\sigma = 0$ on $\supp(\nu) \subset S_0$
	we then also have $J_{\sigma}(\nu) < \infty$.
	Since $\sigma = 0$ and $G\mu^* \leq \varphi - \delta$ on the support of $\nu$, 
	we find
	\begin{align*}
	\int (G\mu^* + \sigma u^* - \varphi) d\nu \leq  - \delta < 0 
	\end{align*}
	which is in contradiction with Lemma \ref{lemma26} (a).
	This proves part (b) of Theorem~\ref{secondthm}. 	
\end{proof}

\subsection{Solution of equation \eqref{dr_soliton_gas2}}\label{sect-sol-dr2}

In this section we use Theorems \ref{solitonmain} 
and \ref{secondthm} to prove the existence of a solution
to \eqref{Gueq} with $\varphi(z)$ that is sufficiently smooth on some neighborhood 
of $\G^+$, but is not necessarily  positive or superharmonic in $\C^+$.
As an  example,  $\varphi(z) = -4 \Im z \Re z$ corresponds to equation \eqref{dr_soliton_gas2}.
We start with the following lemma.

\begin{lemma} \label{lem-pos-def}
The energy functional $J_0(\m)$ with $\varphi\equiv 0$, see \eqref{J0}, is non negative on the set of 
all signed compactly supported Borel measures on $\C^+$.  Moreover, it is zero if and only if $\m=0$.
\end{lemma}

\begin{proof}
Let $\G^+\subset \C^+$ be a compact containing $\supp \m$
Denote $\G:=\G^+\cup \G^-$, where the compact
 $\G^- \subset \C^-$ is Schwarz symmetrical to $\G^+$. 
Any (signed) Borel measure  $\m$ with $\supp \m\subset \G^+$ we anti-symmetrically extend to the signed measure
$\tilde \m$ on $\G$ by setting $\tilde \mu(S)=\mu(S)$ if $S\subset\G^+$ and $\tilde \mu(S)=-\mu(S)$
 if $S\subset\G^-$. 
 Then  $\tilde \mu(\G)=0$ and  it is easy to check that
   \begin{equation} \label{log-pot-sym}
   G\mu(z)=-\frac{1}{\pi } \int\log|z-w|d\tilde{\m}(w)=:U\tilde \mu(z),
   \end{equation}
   where $U\tilde \m$ is the standard logarithmic potential of $\tilde \m$.
   Now the statement follows from \cite{SaffTotik}, Lemma I.1.8. 
 \end{proof}

 Lemma \ref{lem-pos-def} implies that  the operator $v \mapsto Gv$ defined by \eqref{Greenpot}
 (in some suitable function space)
is positive definite. Therefore,
if $\sigma > 0$
on $\Gamma^+$ then the operator $v \mapsto Gv +\sigma v$
is also positive definite, with the spectrum bounded away from $0$; then its inverse exists 
and $v = (G + \sigma)^{-1} \varphi$ is the solution. 
If $\sigma$ has zeros on $\Gamma^+$ then this argument
does not work. However, 
 Theorem \ref{theo-eq2} stated below covers the latter case.

\begin{theorem} \label{theo-eq2}
If  $\varphi$ is a $C^2$ function in a neighborhood
of $\Gamma^+$ then equation \eqref{Gueq}
has a unique weak solution (in the sense of Theorem \ref{solitonmain}).  
If either $\sigma \equiv 0$ on $\Gamma^+$
or $\sigma > 0$ on $\Gamma^+$, then \eqref{Gueq}
is valid everywhere on $\Gamma^+$, and the solution is
a  regular solution. 
\end{theorem}
\begin{proof}
Let $h$ be a compactly supported
$C^2$ function in $\C$ with $h=\varphi$ on $\Gamma^+$ and $h=0$ on $\R$.
Let $K \subset \mathbb C^+$ be a  compact set containing  $\G^+$ as well as the 
support of $h$.
Since $h$ is $C^2$ and has compact support there is $c > 0$
such that $\Delta h < 2c$ in $\C$, where $\Delta$ is the Laplace operator. 

Let $w$ be a non-negative $C^2$ function in $\C^+$ with compact support
such that $w = c$ on $K$. Then the function
\[ \varphi_1(z) =
	\frac{1}{\pi} \int_{\mathbb C^+}
		\log \left|\frac{z- \bar{s}}{z-s} \right| w(s) dA(s) \]
where $dA$ is planar Lebesgue measure, 
is superharmonic and $\Delta \varphi_1 = -2w$.
So,
\[  \Delta( h + \varphi_1) = \Delta h +  \Delta \varphi_1
	< 2c - 2w = 0 \qquad \text{on } K, \]
while $\Delta (h+ \varphi_1) =  \Delta \varphi_1 = - 2w \leq 0$ on $\mathbb C \setminus K$. 

Then we split
\[ \varphi =  h + \varphi_1 - \varphi_1 \qquad \text{on } \Gamma^+ \]
where both $\varphi_1$ and $\varphi_2 = h+\varphi_1$ 
are continuous and superharmonic on $\mathbb C^+$.
Both functions are non-negative, and since they are not-identically zero, they  must be positive on $\mathbb C^+$ by the
minimum principle for superharmonic functions.

Thus Theorem \ref{solitonmain} applies, and there are positive measures $\mu_1^*$
and $\mu_2^*$ on $\Gamma^+$ with corresponding densities
$u_1^*$ and $u_2^*$ with respect to $\lambda$, 
such that for $j=1,2$,
\begin{align}\label{eq-split}
G\mu_j^* + \sigma u_j^* = \varphi_j,  
\end{align}
in the weak sense (i.e., $\mu_j^*$-a.e., and equality
outside the support of $\mu_j^*$)
Hence
\[ G(\mu_2^* - \mu_1^*) + \sigma (u_2^*-u_1^*)
	= \varphi_2 - \varphi_1 = \varphi \quad \text{ on } \Gamma^+ \]
in the weak sense.	
	Thus, the signed measure $\mu_2^*- \mu_1^*$ solves \eqref{Gueq}, where $\sigma u$ term is understood in the same way as in 
	Theorem \ref{solitonmain}.
	
If $\sigma \equiv 0$ on $\Gamma^+$
or $\sigma > 0$ on $\Gamma^+$, then Theorem \ref{secondthm}
applies and the equations \eqref{eq-split}
are satisfied everywhere on $\Gamma^+$.
Then $\mu_2^*-\mu_1^*$ solves \eqref{Gueq} on $\Gamma^+$.
Finally, the uniqueness of solution follows from Lemma \ref{lem-pos-def}.
\end{proof}

\section{Properties of the minimizer $\mu^*$ under additional assumptions on $\Gamma^+$ and $\sigma$  }\label{sect-prop-mu}

In this section we study the support of the minimizer $\mu^*$ and its  smoothness under some additional assumptions.
Everywhere in this section we assume that  $\Gamma^+$ is a finite union of compact 1D arcs and 2D  closed regions equipped with the standard
Lebesgue measure $\lambda$ each. Moreover,  unless specified otherwise, we assume that $\varphi$ is such that 
equation \eqref{Gueq} has a weak solution.

\subsection{Geometry of $\supp \m^*$}

We start with the following lemma. 

\begin{lemma} \label{lem-s=0}
	Suppose $\Gamma^+$ contains a simple closed contour 
	that is the boundary $\partial \Omega$ of a 
	bounded open set $\Omega$. 
	Let $\mu^*$ be a solution of \eqref{Gueq}: 
	$G\mu^*+\sigma u^*=\varphi$   with some $\sigma \geq 0$ on $\Gamma^+$.
Assume that $\sigma =0$ on $\partial \Omega$ and 
$\varphi$ is harmonic on $\Omega$. 
Then $\supp \mu^*\cap \Omega = \emptyset$,
i.e., there is no support of $\mu^*$ inside $\Omega$.
\end{lemma} 
\begin{proof}
	Since $G \mu^*$ is superharmonic with
	$G\mu^*  = \varphi$ on $\partial \Omega$
	and $\varphi$ harmonic on $\Omega$,
	the minimum principle tells us that $G \mu^* \geq \varphi$
	in $\Omega$. Since $G \mu^* \leq G\mu^* + \sigma  u^* =  \varphi$ on
	$\Gamma^+$, it follows that $G \mu^* = \varphi$
	on $\Omega \cap \Gamma^+$.

	Then $G\mu^*$ is harmonic on $\Omega \setminus \Gamma^+$ (since $\mu^*$ is supported on $\Gamma^+$)
	and it agrees with $\varphi$ on its boundary. Then by the
	maximum/minimum principle for harmonic functions we get
	$G \mu^* = \varphi$ on $\Omega$ and $G \mu^*$ is
	harmonic on $\Omega$.   
	Since (in distributional sense) $\Delta G \mu^* =
	-2 \mu^*$, we then conclude that $\mu^* = 0$ on $\Omega$. 
\end{proof}

 Lemma \ref{lem-s=0} leads to some interesting consequences for equation \eqref{dr_soliton_gas1},
 where $\varphi(z)=\Im z$ is harmonic in $\C$. The most obvious is that in the case $\s\equiv 0$
 on a compact connected region $\G^+$ then $\supp \m^*\subset \part \G^+$. This is a well known fact 
 in potential theory. Moreover, using Proposition \ref{propGvareq2}, we obtain that $\supp \m^*$
 coincides with the outer boundary of $\G^+$. In fact, the latter result holds even if $\s=0$ only
 on the outer boundary of $\O$.
 
 Consider another case when equation \eqref{dr_soliton_gas1} has a solution $u$ on 
 the compact domain $\G^+$ and $\s=0$ on a simple closed curve $\g\subset\G^+$. Then, according to 
 Lemma \ref{lem-s=0}, if $u$ is bounded  on $\g$ then $u\equiv 0$ inside the region bounded by $\g$.

\begin{example}\label{ex-circle}
	Consider the example of circular condensate from \cite{ElTovbis}, where $\G^+$ consists of the 
	upper semicircle $|z|=\r$, $\Im z\geq 0$, with some $\r>0$ and $\s\equiv 0$ on $\G^+$. Then
	$u=\frac{\Im z}{\pi \r}$ and $v=\frac{-4\Im z^2}{\pi \r}$ are solutions of 
	\eqref{dr_soliton_gas1}-\eqref{dr_soliton_gas2} respectively. If we replace $\G^+$ by the upper
	semi disk $|z|\leq \r$, $\Im z\geq 0$ and let $\s$ to be any positive continuous function on $\G^+$
	such that $\s(z)=0$ when $|z|=\r$, then the same $u,v$ on the upper semi circle $|z|=\r$, $\Im z\geq 0$
	with trivial continuation $u=v\equiv 0$ for $|z|<\r$ solve \eqref{dr_soliton_gas1}-\eqref{dr_soliton_gas2} 
	respectively. This example, strictly speaking, does not satisfy conditions of  Lemma \ref{lem-s=0} since
	$\G^+\cap\R\not=\emptyset$ but, nevertheless, it illustrates the idea.
\end{example}

  \subsection{Smoothness in 1D  case with $\s\equiv 0$} \label{sec-sig=0-smooth}
 
 In the rest of this section we consider the smoothness of $d\m^*$ in the 1D case, i.e., when $\G^+$ is a finite collection
 of piece-wise smooth curves that can be closed or opened. 
 Transversal intersection of different curves are allowed, but we consider intersection points, as well
 as the end points of open arcs, as points of non smoothness.
 In this case the reference measure $\lambda$ is simply the arclength measure on the curves.
We also assume that $\varphi$ is a $C^\infty$ function in some neighborhood containing $\G^+$.

We start by considering the case of a soliton condensate, i.e., the case of $\s\equiv 0$ (on $\G^+$).
There are certain results about the smoothness of the minimizing measure $\m$ on a 1D compact $\G^+$ in terms of  its 
logarithmic potential $U\m = -\frac{1}{\pi }\int_{\G^+}\log|z-w||dw|$ in the literature, see, for example, 
\cite{SaffTotik}.  Many of these results can also be applied to signed measures $\m$. It turns out 
that these results can be applied to the Green potential $G\m^*$ of the minimizer $\m^*$ of $J_0$.
 
 Indeed, denote $\G:=\G^+\cup \G^-$, where the compact
 $\G^- \subset \C^-$ is Schwarz symmetrical to $\G^+$. 
Any  Borel measure  $\m$ with $\supp \m\subset \G^+$ we anti-symmetrically extend to the signed measure
$\tilde \m$ on $\G$ by setting $\tilde \mu(S)=\mu(S)$ if $S\subset\G^+$ and $\tilde \mu(S)=-\mu(S)$
 if $S\subset\G^-$. The function $\s$ is extended to $\G$ Schwarz symmetrically. Then, see 
 \eqref{log-pot-sym}, $G\m=U\tilde \m$,   so that 
$   G\m+\s u=\varphi$ on $\G^+$ if and only if
\begin{equation}
\label{full-eq1}
-\frac{1}{\pi } \int_\G \log|z-w|d\tilde{\m}(w)+ \s(z) \tilde u(z)= \tilde\varphi (z)
\qquad \text{ on $\G$},
\end{equation}
where $\tilde \varphi$ denotes anti Schwarz symmetrical (odd) continuation of the function $\varphi$ from $\G^+$ to $\G$
and $\s\tilde u$ is the density of $\s\tilde \m$. 

In the case $\s\equiv 0$ on $\G^+$, considered now, we can WLOG (see Lemma \ref{lem-s=0}) assume that
 $\G^+$ belongs to the 
boundary of the 
 unbounded component $\O$ of $\C^+\setminus\G^+$, i.e., 
 $\G^+\subset \part \O$.   Then, by Proposition \ref{propGvareq2}, $\supp\m^*= \G^+$, provided $\varphi>0$ and superharmonic 
 on $\C^+$. Take a smooth subarc
$\g\subset \G^+$ - it is enough to assume that $\g$ is  $C^{1+\d}$  smooth with some $\d>0$. 
Then, according to Theorem II.1.5 of \cite{SaffTotik}, if the logarithmic potential $U\m$ is $\text{Lip}\, 1$ in a 
neighborhood of $\g$,   then $\m$ is absolutely
 continuous on $\g$  and its density $u$ is given by
 \begin{equation} \label{norm-der}
u(s) = \frac{d\m}{ds}(s)=-\frac 1{2}\le(\frac{\part U \m}{\part n_+}(s) + \frac{\part U \m}{\part n_-}(s)\ri),
 \end{equation}
 where $\frac{\part }{\part n_\pm} $ denote two (opposite) normals to $\g$,  and $s$ is the arclength parameter on $\g$.
It is clear that, because of \eqref{full-eq1}, we can replace $U\m$ with the Green potential $G\m$ in the above 
statement. 

Equations \eqref{full-eq1} and \eqref{norm-der}
show that the smoothness of $u$ can be derived from the smoothness of $G\m$
at the boundary $\G^+$ of $\C^+\setminus\G^+$. Since Green potential $\n(z)=Gu(z)$ satisfies the
Dirichlet 
 boundary value problem
\begin{align}
\begin{cases} 	\label{dir-for-v}
\text{ $\n$ is harmonic on $\mathbb C^+ \setminus \Gamma^+$}, \\
\text{ $\n = \varphi$ on $\Gamma^+$}, \\
\text{ $\n  = 0$ on $\mathbb R$ and at infinity},
\end{cases} 
\end{align}
we can use regularity theorems for boundary value problems from classical PDEs to estimate the smoothness of $\n$. 
Under our assumptions we can represent
\begin{equation} \label{decomp}
\C^+\setminus\G^+=\Omega \cup 
\left(\bigcup_{j=1}^k D_j\right),
\end{equation}
where $\O$ denotes the unbounded and $D_j$'s denote bounded simply connected  components. 
All the domains in \eqref{decomp} are disjoint.

Let $\G_1\subset \G^+$ be a $C^\infty$ closed curve that is the boundary of $D_1$. Then 
$\nu$ is harmonic in $D_1$ and satisfies the Dirichlet condition $\n = \varphi$ on $\partial D_1=\Gamma_1$. Assume that
$\Gamma_1$ is positively oriented. Then, using regularity theorems, in particular Schauder Estimates in H\"{o}lder spaces, 
see for example \cite{Egorovshubin}, 
Section II.2.14, or \cite{Taylor}, we obtain that $\n(z)$ exists and is a $C^\infty$ function on the closure of $D_1$, i.e., on $D_1\cup\G_1$.
Thus, partial derivatives of $ G u(z)$ are $C^\infty$ functions on $D_1\cup\G_1$. In particular, so is 
$\frac{\part G u}{\part n_+}(z)$. 

 In the case when  $\G_1=\partial D_1$ is only a piece-wise $C^\infty$  curve,  $\n$ is  $C^\infty$ smooth
everywhere on $\G_1$ except the points of non smoothness (since smoothness of $\nu$ is a local property).
Similar arguments show that $\n$ is  $C^\infty$ smooth
everywhere on the closure of $\Omega$ except points of non smoothness of $\G^+$.
Thus,  $\frac{\part G u}{\part n_-}(z)$ is also a $C^\infty$ function on $\G^+$ away from the points of non 
smoothness
and, according to \eqref{norm-der}, so is $u(s)$. This result is summarized in the following theorem.

\begin{theorem}\label{cor-smooth1D}
Let $\G^+\subset\C^+$ be a piece-wise $C^\infty$ smooth  curve and let $\varphi$ be a  $C^\infty$ function 
in some domain $B$ containing $\G^+$. 
Then the solution $u^*$ of the integral equation \eqref{Gueq}
is a $C^\infty$ function on
any compact subarc of a smooth arc of $\G^+$. 
\end{theorem}
\begin{proof}
We first consider the case when $\varphi>0$ and superharmonic in $\C^+$.
As it was shown above,  both  normal derivatives $\part \n / \part n$ are $C^\infty$ functions on
any compact subarc $\g$ of a smooth arc of $\G^+$ together with the adjacent   regions of $\C^+\setminus \G^+$.
Then, by 
formula \eqref{norm-der}, $u^*=\frac{d\m^*}{ds}$, where $\mu^*$ is the minimizer of the energy functional $J_0$,
is a  $C^\infty$ function on $\g$.
 
Consider now the case of general $\varphi$. As in the  proof of Theorem \ref{theo-eq2}, we can represent 
 $\varphi=\varphi_2-\varphi_1$, 
 where both $\varphi_{1,2}$ are positive and superharmonic in $\C^+$. Moreover, the $C^\infty$ smoothness of  $\varphi$
 implies that  both $\varphi_{1,2}$ are also $C^\infty$  functions. 
 We can now apply arguments from the first part of the proof to obtain the statement of the theorem.
 \end{proof}
 
%##############################

\begin{remark} \label{rem-endp-t}
The case when $\G^+$ is piecewise $C^\infty$ collection of contours  includes the case when $\G^+$ contains arcs
 with endpoints (not closed curves). In that case behavior near the endpoints is given by Theorem IV.2.6 of \cite{SaffTotik},
 where $\a=2\pi$. 
 %Choosing $p=2$, 
 In particular, we obtain that $u(z)|z-z_0|^\epsilon \in L^2_{loc}$ for any $\epsilon>0$ on a piece of $\G^+$ that includes 
 an endpoint $z_0$.
\end{remark} 
 
 \begin{remark}
Let $\G^+$ be a  piece-wise finitely smooth curve (with sufficient smoothness). Then one can use similar arguments
to show that $u^*=\frac{d\m^*}{ds}$ is also piece-wise  finitely smooth. 
\end{remark}

\subsection{Smoothness in 1D  case with $\s\geq 0$} 

Consider now the case when $\s\geq 0$ on $\G^+$.
Take some smooth closed subarc $\g$ of $\G^+$ (it contains its endpoints or encircles a region).  
We assume $\s$ and $\varphi$ to be sufficiently smooth on $\g$ and also 
that $\s>0$ on $\g$.
According to Theorem \ref{secondthm}, the density $u^*=d\m^*/ds$  of the minimizer is continuous on $\g$.
Let us write the 
equation \eqref{Gueq} as
\begin{equation} \label{bv-on-1}
 G_1u^*+\s u^*=\varphi - G_2 u^*  \qquad \text{on $\g$,}
 \end{equation}
 where the integration in $G_{1}$ is over $\g$ and $G_{2}=G-G_{1}$
We can perceive  $u^*$ on $\G^+\setminus\g$ as to be given 
and consider \eqref{bv-on-1} as a second kind Fredholm     integral equation for $u^*$ on $\g$. The right hand side of 
\eqref{bv-on-1} is harmonic at the interior points of $\g$.
It can also be shown to be  H\"{o}lder continuous
on $\g$ with any H\"{o}lder exponent $\nu\in(0,1)$.
Under these conditions it follows, 
see, for example, \cite{Muskhelishvili}, Section 51.1,
that $u^*$ is H\"{o}lder continuous
on $\g$ with any H\"{o}lder exponent $\nu\in(0,1)$. 

Let us  differentiate \eqref{bv-on-1} with respect to the arclength $s$. We obtain
\begin{equation} \label{bv-on-1a}
\s (u^*)'= -\s' u^*- (G_1 u^*)'+ \varphi' - (G_2 u^*)' =: RH.
\end{equation}
We now prove that the right hand side of \eqref{bv-on-1a} is H\"{o}lder continuous
on a proper compact  subarc  $\g_0$ of
$\g$. It is obvious that the first and the third terms of $RH$ are H\"{o}lder continuous
on $\g$.  
The second term $(G_1 u^*)'$ is a singular
integral of  a H\"{o}lder continuous function and therefore
must be H\"{o}lder continuous on $\g$ with the same $\nu$. Finally,
since $u^*$ is  H\"{o}lder continuous on $\g$,
$(G_2 u^*)'$ is  harmonic on $\g_0$.
 So, we proved the following lemma.

\begin{lemma} \label{lem-sig-1}
Let $\g\subset \G^+$ be a smooth closed arc such that $\s>0$ and is smooth on $\g$. Then $\frac{du^*}{ds}$  is  H\"{o}lder continuous
on any compact subarc $\g_0\subset\g$ with any H\"{o}lder exponent $\nu\in(0,1)$. Here $s$ is the arclength parameter  on $\g$.
\end{lemma}

 \begin{lemma} \label{lem-sig-0}
 Let $\s=0$ on some open $C^\infty$  smooth arc $\G_1\subset\G^+$.
 Then $u^*=d\m^*/ds$ is $C^\infty$ smooth  of $\G_1$. 
 \end{lemma}
\begin{proof}
We write $Gu^*=G_1u^*+G_2u^*$,
where  in $G_j u^*$ we integrate over $\G_j$, $j=1,2$   and $\G_2=\G^+\setminus \G_1$. Then
 \begin{equation} \label{bv-on-01}
 G_1u^*(z)=\varphi(z) - G_2 u^*(z)
 \end{equation}
 on $\G_1$. Let $B$ be a region containing $\G_1$ and separated from $\R$.
 Take a function $\phi$, harmonic in $B\setminus \G_1$, continuous in the closure of $B$ and satisfying 
 $\phi=\varphi(z) - G_2 u^*(z)$ on $\G_1$. Denote by $f$ a $C_0^\infty$ extension of $\phi$ to the closure
 of $\C^+$ satisfying $f_\R=0$.
 %Let $\g_0$ be any compact subarc of $\g$. Note that the right hand side is a harmonic function on $B$.
 Denote $w:=G_1u^*(x)$. Then $w$ satisfies the Dirichlet problem for the Laplace equation in $\C^+\setminus \G_1$, with the boundary
 values $w=0$ on $\R$ and $w=f$ on $\G_1$.  We can now apply Theorem 
 %\ref{theo-reg-Om} and Corollary 
 \ref{cor-smooth1D}
 to prove that $u^*$ is $C^\infty$ smooth  on any compact subarc of $\G_1$.
 \end{proof}

\section{Bound state fNLS  and KdV condensates }\label{sec-quasimom-vert-cond}

Connection between a bound state fNLS soliton gas and the corresponding KdV soliton gas was described in
Section \ref{sec-quasimom-vert}. That connection implies that all the obtained above results about fNLS soliton gases,
applicable to bound state gases, can be reformulated for KdV soliton gases. That includes existence and uniqueness
of solutions, non negativity of the density of states $u$, smoothness and geometry of $\supp u$.
In particular,
the main
 Theorem \ref{solitonmain} and Theorem \ref{secondthm}, are applicable to the KdV soliton gas NDR
 \eqref{dr_kdv-soliton_gas1}-
 \eqref{dr_kdv-soliton_gas2} with a given continuous and non negative on $\G^+$ function $\s(z)$. 
 
Moreover, it turns out that the solution $u$ of \eqref{dr_soliton_gas1} in the case of the bound state fNLS condensate 
($\s\equiv 0$) is proportional to the density of the quasimomentum differential $dp$, see \eqref{pq},
on the hyperelliptic Riemann surface
$\Rscr$ associated with $\G$. This result is formulated 
(Theorem \ref{lem-boun-cond}) 
and proven in this section. Its extension to the KdV condensate is also addressed below.   

Consider $\s\equiv 0$ and $\G=\G^+\cup\G^-\subset i\R$, which is Schwarz symmetrical and consists of $2N+1$ segments with 
endpoints $ib_j$, $j=0, 1\dots,N$,
and beginning points $ia_j$, $j=1\dots,N$, in  $\C^+$, where $0<b_0<a_1<b_1<\dots<b_N$, and their complex conjugates in $\C^-$.
This is the case of a general (even genus)  bound state fNLS condensate
mentioned in Section \ref{sec-quasimom-vert}, which has $v(z)\equiv 0$ solution to \eqref{dr_soliton_gas2}.
Our goal is the following theorem.

\bt \label{lem-boun-cond}
Denote by $dp$ the real normalized quasimomentum differential on the Riemann surface $\Rscr$ (see Section \ref{sect-backg}).
Then: i) $dp$ has zero  $\bf{B}$ - periods; ii) $\frac{dp}{dz}$ is    Schwarz symmetrical (odd on $i\R$); 
iii) $u(z)=\frac{idp}{\pi dz}>0$ on $\G^+\setminus \{0\}$ and it satisfies  \eqref{dr_soliton_gas1} with $\s\equiv 0$
on $\G^+$.
\et

\begin{proof}
Switching from the condensate equation \eqref{dr_soliton_gas1} to equation \eqref{full-eq1} (with anti Schwarz symmetrical $u$) and 
differentiating  the latter 
 in $z$,  
we obtain $\pi H u=1$ on $\G$, where $H$ denotes the Finite Hilbert Transform (FHT) on $\G$ (which is oriented upwards). 
The inversion formula for FHT $H$ (see, for example, \cite{OkadaElliott} when $N=0$) yields
\begin{align}\label{u-eq} 
u(z)=\frac{-1}{\pi^2 R(z)}\int_{\G}\frac{R_+(w)dw}{w-z}= \frac{-1}{2\pi^2 R(z)}\oint_{\hat\g}\frac{R(w)dw}{w-z}=\cr
\frac{i}{\pi R(z)}\le(\le. \Res \frac{R(w)}{w-z}\ri|_{w=\infty} -\k\le. 
\Res \frac{R(w)}{w-z}\ri|_{w=z}
%\Res R(w)\ri|_{w=z}
\ri)
\end{align}
where  $\hat\g$ is a negatively oriented circle containing $\G$ but not
containing $z$ if $z\not\in\G$ and $\k=0$ if $z\in\G$ with $\k=1$ otherwise.
Calculating the residue at $w=\infty$ we obtain
\begin{equation} \label{u-exp}
u(z)=\frac{iP(z)}{\pi R(z)}\qquad \text{on $\G$,}
\end{equation}
where $P(z)$ is a monic odd polynomial of degree $2N+1$ with real coefficients. The exact values of these coefficients, which can be obtained 
from $\le.\Res \frac{R(w)}{w-z}\ri|_{w=\infty}$, are not essential because  the null space of $H$ is spanned by
$\frac{z^k}{R(z)}$, where $k=0,\dots,2N-1$. Since $u(z)$ must be odd (anti-Schwarz symmetrical), we consider only the odd 
powers of $k$. It is clear that $u(z)\in\R$ when $z\in \G$. Note that
 $ R(z)$ has opposite signs on the neighboring segments of $\G$,
 where by convention  we evaluate $R$ on the positive (left) side of $\G$. 
It is clear that in order to have $u>0$ on $\G^+$  and $u<0$ on $\G^-$  the polynomial  $P(z)$ must have a zero in each of the $2N$ gaps between  
consecutive  segments of $\G$ (the remaining zero is  $z=0$).
To determine these zeros 
%(or the  coefficients) 
of $P(z)$ we use the fact that the logarithmic potential   
\begin{equation} \label{Gu2N-1}
Gu(z)=-\frac 1\pi\int_{\G}\log \le| w-z\ri|\frac{iP(w)(-idw)}{R(w)}
\end{equation}
must be continuous, see Theorem I.5.1, assertion 4, \cite{SaffTotik}. Differentiating \eqref{Gu2N-1} in $z$
and using the residues we obtain 
\begin{equation} \label{Gu2N}
[Gu(z)]'=\frac 1\pi\int_{\G}\frac{P(w)dw}{(w-z)R(w)}=-i\le (1- \k\frac{P(z)}{ R(z)}\ri),
\end{equation}
so that
\begin{equation} \label{Gu2N-2}
Gu(z)=-iz+i\k\int^z_m\frac{P(w)dw}{ R(w)},
\end{equation}
where $\k$ is the same as in \eqref{u-eq},  $m=ib_{j-1}$ if $z$ is on the $j$th gap $(ib_{j-1},ia_j)$, $j=1,\dots,N$ and  $m=-ib_{j-1}$ if $ z$ is on the 
complex conjugate gap $-j$ in $\C^-$. 
Now the continuity of $Gu$ requirement is translated into the system of linear equations
\begin{equation} \label{eq-for-P}
\int^{a_j}_{b_{j-1}}\frac{P(w)dw}{ R(w)}=0,\quad   j=1,\dots,N, 
\end{equation}
for the coefficients of the odd monic polynomial $P$. By the symmetry, the corresponding equations hold on the gaps in $\C^-$.

The system \eqref{eq-for-P} has a unique real solution. That follows from that fact that 
$\Im \t$, where $\t$ 
is the Riemann period matrix for $\Rscr$, is positive definite.
Thus, $iudz$ is a real normalized meromorphic differential with the poles at infinity of both sheets,
and, according to \eqref{pq},  $dp=-i\pi u(z)dz$ is the quasi momentum differential on  $\Rscr$.
Moreover, all  the $\bf{B}$ periods of  $dp$ are zeros. 

It remains only  to prove that $u(z)>0$ on $\G^+$. In fact, since the system \eqref{eq-for-P} requires that
there must be just one zero of an odd polynomial
$P(z)$ in every gap, it is sufficient to prove that $u>0$ on the last segment $(ia_N,ib_N)$.  

Indeed, 
since all the zeros of $P(z)$
are on $(-ib_N,ib_N)$, $\arg u(z)=\frac \pi 2$ on $(ib_N,+i\infty)$. When $z$ crosses $ib_N$ and stays on the positive
(left) side of $(ia_N,ib_N)$, the argument $of R(z)$ gains $\frac \pi 2$ whereas the argument of $P(z)$ does not change.
Hence, the lemma is proved.  
\end{proof}

\begin{remark} \label{rem-odd-gen}
The above arguments can be repeated for Schwarz symmetrical  $\G\subset i\R$ that consists of  $2N$ segments.
The corresponding $\Rscr$ has   genus $2N-1$ and $u(z)=\frac{iP(z)}{\pi R(z)}$ {on $\G$,} where $R(z)$ is odd
and $P(z)$ is even. $P(z)$ must have exactly one zero in each of the $2N-2$ gaps lying entirely in $\C^+$ or $\C^-$
and exactly two symmetrical zeros in the central gap $[-ia_1, ia_1]$.
\end{remark}

Theorem \ref{lem-boun-cond} and Remark \ref{rem-odd-gen} imply that solutions of the NDR \eqref{dr_kdv-soliton_gas1}-
 \eqref{dr_kdv-soliton_gas2}  for the KdV soliton condensate   ($\s\equiv 0$) is always represented by the density
 of the corresponding meromorphic differentials on $\Rscr$. In particular, $u_{KdV}(z)=\hf u_{fNLS}(iz)$, $z\in \G^+$,
 where $\G^+\subset \R$ and $u_{KdV}$ are defined for \eqref{dr_kdv-soliton_gas1} and $u_{fNLS}$ is defined by
 the corresponding \eqref{dr_soliton_gas1}. Equation \eqref{dr_kdv-soliton_gas2} can be solved similarly to 
 \eqref{dr_kdv-soliton_gas1}. Its solution is given by the density of the corresponding real normalized meromorphic differential that
 has $O(z^2)$ behavior as $z\to \infty$ (on both sheets).

\end{document}